%% file: AFU11-Full.tex
\documentclass[11pt,letterpaper]{article}
\usepackage{amsmath,amsfonts,graphicx,verbatim,fullpage}
\usepackage[ruled,noend]{algorithm2e}

\renewcommand{\algorithmcfname}{Algorithm}
\def\prec{<}

\def\OddCycle{\text{OddCycle}}

\def\squareforqed{\hbox{\rlap{$\sqcap$}$\sqcup$}}
\def\qed{\ifmmode\squareforqed\else{\unskip\nobreak\hfil
\penalty50\hskip1em\null\nobreak\hfil\squareforqed
\parfillskip=0pt\finalhyphendemerits=0\endgraf}\fi}

\renewcommand{\algorithmcfname}{Algorithm}
\newtheorem{theorem}{Theorem}[section]

\newtheorem{definition}{Definition}[section]

\newtheorem{lemma}{Lemma}[section]
\newtheorem{example}{Example}[section]
\newenvironment{proof}{\smallskip\noindent{\em Proof.}}{\hfill\qed\medskip}

\title{Enumerating Subgraph Instances Using Map-Reduce%
\thanks{This work was supported by the project Handling Uncertainty in Data Intensive Applications, co-financed by the European Union (European Social Fund - ESF) and Greek national funds, through the Operational Program ``Education and Lifelong Learning'', under the research funding program THALES.}}

\author{Foto N.\ Afrati \\
National Technical Univ.\\ of Athens\\
{\tt afrati@softlab.ece.ntua.gr}
\and
Dimitris Fotakis\\
National Technical Univ.\\ of Athens\\
{\tt fotakis@cs.ntua.gr}
\and
Jeffrey D.\ Ullman\\
Stanford University\\
{\tt ullman@gmail.com}}

\begin{document}
\date{\today}
\maketitle

\input{triangles_full}

\bibliographystyle{plain}
\bibliography{map-reduce}
\end{document}

%% file: triangles_full.tex
\begin{abstract}
The theme of this paper is how to find all instances of a given ``sample'' graph in a larger ``data graph,'' using a single round of map-reduce.  For the simplest sample graph, the triangle, we improve upon the best known such algorithm.  We then examine the general case, considering both the communication cost between mappers and reducers and the total computation cost at the reducers.  To minimize communication cost, we exploit the techniques of \cite{AU10} for computing multiway joins (evaluating conjunctive queries)  in a single map-reduce round.  Several methods are shown for translating sample graphs into a union of conjunctive queries with as few queries as possible.  We also address the matter of optimizing computation cost.
Many serial algorithms are shown to be ``convertible,'' in the sense that it is possible to partition the data graph, explore each partition in a separate reducer, and have the total work at the reducers be of the same order as the work of the serial algorithm. For data graphs of unrestricted degree, we show that there are convertible algorithms whose running time is of the same order as the lower bounds on number of occurrences of the sample graph that were provided by \cite{Alon81}.  We also offer better convertible algorithms when the degree of nodes in a data graph of $m$ nodes is limited to $\sqrt{m}$.
\end{abstract}

\section{Introduction}
\label{intro-sect}

We address the problem of finding all instances of a given subgraph (the {\em sample graph}) in a very large graph (the {\em data graph}). The problem is computationally intensive, so we shall concentrate on algorithms that can be executed by a single round of map-reduce \cite{DeanGhemawat}. We investigate how to minimize two important measures of complexity.
The first is the {\em communication cost}, i.e., how to hash the edges of the data graph to the reducers in order to minimize the total amount of data transferred from the mappers to the reducers. This problem is, in a sense, a special case of evaluating conjunctive queries or multiway joins on a single large relation, so our starting point is the algorithms for optimal evaluation developed in \cite{AU10} and \cite{AU11-Datalog2.0}.
The second measure is the computation cost, which can be increased significantly when
we move from serial algorithms to their parallel implementation. However, our techniques
derive a parallel algorithm of the same complexity as the serial algorithm.

\subsection{Applications}
\label{apps-subsect}

Finding occurrences of particular sample graphs in a social network is a tool for analyzing and understanding such networks. For instance, \cite{Leskovec} shows how the stage of evolution of a community can be related to the frequency with which certain sample graphs appear. Our results apply directly to problems of these types. Similarly, \cite{AlonDHHS08} discusses how discovering instances of sample graphs supports work in Biomolecular networks.

Another example application concerns analysis of networks for discovering potential threats or discovering recommendations. In these applications, the edges of the network are labeled and possibly directed (e.g., ``buys from'' or ``knows''), and the goal is to find sets of individuals with specific interconnections among them (see e.g., \cite{VS-COSI,VS-Budgeting}). For example, aiming to discover potential threats, one may want to answer questions like ``find all instances of five people booked on the same flight each of whom has bought explosive materials in the past three months.'' Our methods can be extended to this sort of problems as well, although there are two relatively simple extensions needed.

\begin{enumerate}
\item The cited papers assume that the query (sample graph) specifies at least one node (individual) of the data graph.  As a result, optimum algorithms for evaluation will surely start by searching from the fixed node or nodes.  However, eventually, this search will lead to a neighborhood that is sufficiently large that sequential search no longer makes sense.  At that point, our methods can take over with what remains of the sample graph after removing nodes that have been explored on the data graph.

\item We assume that edges are unlabeled.  However, a graph with labeled edges can be represented by a collection of relations, one for each label.  Search for instances of a sample graph can still be expressed as a conjunctive query, and the same techniques applied.
\end{enumerate}

A discussion about the importance of using the map-reduce environment for finding subgraph instances in large data graphs can be found in~\cite{Cohen09}.

\subsection{Measures of Complexity}
\label{costs-subsect}

There are two ways to measure the performance of map-reduce algorithms.

\begin{enumerate}
\item {\em Communication cost} is the amount of data transmitted from the mappers to the reducers.  In the algorithms discussed here, edges of the data graph are replicated; i.e., they are associated with many different keys and sent to many reducers.

\item {\em Computation cost} is the total time spent by all the mappers and reducers. In the algorithms to be discussed, the mappers do nothing but assign keys to the input (the edges of the data graph), so their computation cost is proportional to the communication cost. We shall therefore discuss only the computation cost at the reducers in this paper.
\end{enumerate}
These measures and their relationship are discussed in \cite{ABCPU-EDBT11}.

Another measure we address is the ``number of reducers'' used by different algorithms. What we are actually measuring is the number of different keys, and this quantity is an upper bound on the number of Reduce tasks that could be used.  Lowering the number of reducers is not necessarily a good thing, but the communication cost for all the algorithms discussed grows with the number of reducers.  See Section~\ref{tri-compare-subsect}, where we show how algorithms that are parsimonious in their use of reducers can lead to lower communication cost when the number of reducers is fixed.

\subsection{Outline of the Paper and Contributions}
\label{outline-subsect}

This paper is the first to offer algorithms for enumerating all instances of an arbitrary sample graph in a large data graph, using a single map-reduce round.  We combine efficient mapping schemes to minimize communication cost with efficient serial algorithms to be used at the reducers.
Throughout the paper:

\begin{enumerate}
\item The data graph is denoted $G$ and has $n$ nodes and $m$ edges.

\item The sample graph is denoted $S$ and has $p$ nodes.
\end{enumerate}

We first address the communication cost of algorithms for finding all instances of a sample graph in a data graph.  Section~\ref{tri-mj-sect} motivates the entire body of work.  We apply the multiway join algorithm of \cite{AU10} to the triangle-finding problem.  Although multiways joins are frequently more expensive than a cascade of two-way joins, for the problem of finding triangles, or more generally instances of almost any sample graph, the multiway join in a single round of map-reduce is more efficient than two-way joins, each performed by its own round of map-reduce.
Specifically, our result of using multiway joins is an improvement to the one-round algorithm of \cite{SV11}.% mentioned in Section~\ref{rel-work-subsect}.

The balance of the paper deals with arbitrary sample graphs and is divided into two parts. First we look at communication cost beginning in Section~\ref{sg-cq-sect} and then we address computation cost starting in Section~\ref{comp-cost-sect}.

In Section~\ref{sg-cq-sect} we look at arbitrary sample graphs.  We generate from a given sample graph a collection of conjunctive queries with arithmetic constraints that together produce each instance of the sample graph exactly once.  We then use the automorphism group of the sample graph and the collection of edge orientations to simplify this collection, while still producing each instance only once.

Section~\ref{opt-eval-sect} covers the optimal evaluation of the conjunctive queries from Section~\ref{sg-cq-sect}.  We give a simple algorithm for minimizing the communication cost for a single conjunctive query, and then show how it can be modified to allow all the conjunctive queries to be evaluated in one map-reduce round.
We show that combining the evaluation of all conjunctive queries into a single map-reduce job always beats their separate evaluation.

Then, Section~\ref{cyc-sect} looks at the special case of enumerating cycles of fixed length.  We give a method for generating a smaller set of conjunctive queries than is obtained by the general methods of Section~\ref{sg-cq-sect}.

Section~\ref{comp-cost-sect} begins our examination of optimizing the computation cost of map-reduce algorithms.  All the map-reduce algorithms we discuss involve partitioning the nodes and edges of the data graph into subgraphs and then looking for instances of the sample graph in parallel, in each subgraph.  Thus, the key question to address is under what circumstances a serial algorithm for finding instances of a sample graph $S$ will yield a map-reduce algorithm with the same order of magnitude of computation.  We call such an algorithm {\em convertible}  and give a (normally satisfied) condition under which an algorithm is convertible.

It turns out that all sample graphs have convertible algorithms.  The real question is what is the most efficient convertible algorithm.  The results of \cite{Alon81} give worst-case lower bounds for the running time of serial algorithms, and we shall show in Section~\ref{serial-alg-sect} that these lower bounds can be met by convertible algorithms (Theorem \ref{gen-enum-th}) . Moreover, when we limit the degree of nodes in the data graph, we can obtain more efficient, yet still convertible, algorithms (Theorem \ref{bounded-degree-th}).

\subsection{Related Work}
\label{rel-work-subsect}

In the 1990's there was a considerable effort to find good algorithms to (a) detect the existence of cycles of a given length and/or (b) count the cycles of a given length \cite{AYZ97}.  A generalization to other sample graphs appears in \cite{KLL11}. These problems reduce to matrix multiplication.  However, enumeration of all instances of a given subgraph cannot be so reduced.  Probabilistic counting of triangles was discussed in \cite{CTKMF09}.
More recently, there has been significant interest in probabilistic (approximate) counting of small sample graphs on large biological and social networks.  To this end, Alon et al.\ \cite{AlonDHHS08} applied the color-coding technique, and obtained a randomized approximation algorithm for counting the occurrences of a bounded-treewidth sample graph with $p$ nodes on a data graph with $n$ nodes in $O(2^{O(p)} n)$ time. Subsequently, Zhao et al.\ \cite{ZhaoKKM10} showed how the approach of \cite{AlonDHHS08} can be parallelized in a way that scales well with the number of processors.

Enumeration of triangles has received attention recently.   It was the subject of the thesis by Schank \cite{Schank07}.
Suri and Vassilvitskii \cite{SV11} give one- and two-round map-reduce algorithms for finding triangles. In this paper, we begin with an improvement to their one-round algorithm, obtained by the use of multiway joins, in Section~\ref{tri-mj-sect}, and then give the extensions needed for arbitrary subgraphs.

In \cite{PaghT11} the triangle finding problem in map reduce is experimentally studied; actually, this paper implements a randomized counting algorithm for triangles. Finally, in a related problem, a few papers have investigated recently the question of counting the output of a multiway join \cite{AtseriasGM08} and finding a serial algorithm for computing optimally a multiway join \cite{NgoPRR12}. The complexity of the algorithm presented in \cite{NgoPRR12} is the same as the worst case maximum size of the output of a multiway join \cite{AtseriasGM08}. The upper bound of \cite{AtseriasGM08} on the output size of a multiway join is obtained as the solution of a linear program, and is essentially tight, in the sense that for infinitely many sizes of input relations, there exists a multiway join instance with the prescribed output size.
On the other hand, one can show that the bound of \cite{AtseriasGM08} is arbitrarily bad at an infinite number of combinations of relation sizes. However, since our joins use only a single relation, we match exactly the bound of \cite{AtseriasGM08} for all sizes, except for matters of ``rounding errors.''

\section{Triangles and Multiway Joins}
\label{tri-mj-sect}

In this section we see that the problem of finding triangles in a large graph using a single round of map-reduce is a special case of computing a multiway join.  We begin by discussing the ``Partition Algorithm,'' which is a recent idea that almost-but-not-quite implements a multiway join.  Then, we show how to apply the technique of \cite{AU10} for optimal implementation of multiway joins by map-reduce.  Finally, we combine these ideas with those of Partition to get a method that works better than either.

\subsection{The Partition Algorithm of Suri and Vassilvitskii}
\label{partition-subsect}

\cite{SV11} gives the {\em Partition Algorithm} for enumerating triangles using a single round of map-reduce.  This method has the property that the total work done by the mappers and reducers is no more than proportional to the work that would be done by a serial algorithm for the same problem.  However, as we shall show in Section~\ref{tri-compare-subsect}, the communication cost of Partition is almost, but not quite as good as one can do.

Partition works as follows.
Given a data graph of $n$ nodes and $m$ edges, partition the $n$ nodes into $b$ disjoint subsets $S_1,S_2,\ldots,S_b$ of equal size.  For each triple of integers $1\le i<j<k\le b$ create a reducer $R_{ijk}$; thus there are
$\binom{b}{3} = b(b-1)(b-2)/6$
or approximately $b^3/6$ reducers.  The mappers send to each $R_{ijk}$ those edges both of whose nodes are in $S_i\cup S_j\cup S_k$.  Thus, each reducer has a smaller graph to deal with; that graph has $3n/b$ nodes and an expected number of edges $m/b^2$.  The paper \cite{SV11} shows that assuming a random distribution of the edges, the total work of the mappers and reducers is $O(m^{3/2})$, which is also the running time of the best serial algorithm \cite{Schank07}.

The communication cost for Partition can be calculated as fol\-lows.  An expected fraction $1/b$ of the edges will have both their ends in the same partition, say $S_i$.  This edge must be sent by the mappers to $\binom{b-1}{2} = (b-1)(b-2)/2$ of the reducers~-- the reducers corresponding to all the subsets of three integers that includes $i$.  The remaining fraction $(b-1)/b$ of the edges have their ends in two different partitions, say $S_i$ and $S_j$.  These edges are sent to only $b-2$ reducers, those corresponding to the subsets of the integers that include both $i$ and $j$.
The total communication per edge between the mappers and reducers is thus
$$\frac1b (b-1)(b-2)/2 + \frac{b-1}{b} (b-2) = \frac32(b-1)(b-2)/b$$
For large $b$, the total communication cost for all the edges is approximately $3bm/2$.

As we shall see, the small problem with the partition algorithm is the fact that some edges need to be copied too many times.  This problem also shows up in the details of the algorithm, where certain triangles~-- those with an edge both of whose ends are in the same partition~-- are counted more than once, and the algorithm as described in \cite{SV11} needs to do extra work to account for this anomaly.  In the variant we propose in Section~\ref{improved-subsect}, all edges are replicated the same number of times, and the communication cost is lowered from $3b/2$ to $b$ per edge.

\subsection{The Multway-Join Algorithm}
\label{mj-subsect}

In \cite{AU10} the execution of multiway joins by a single round of map-reduce was examined, and it was shown how to optimize the communication cost.  In fact, the case of finding triangles was considered in the guise of computing a simple cyclic join $R(X,Y)\bowtie S(Y,Z)\bowtie T(X,Z)$.  In the case that the edges are unlabeled (the only case we consider here), the relations $R$, $S$, and $T$ are the same; let us call it $E$.  Then enumerating triangles can be expressed as evaluating the join
$E(X,Y)\bowtie E(Y,Z)\bowtie E(X,Z)$.

There is an important issue that must be resolved, however: does an edge $(a,b)$ appear as two tuples of $E$ or as only one?  If we use tuples $E(a,b)$ and $E(b,a)$, then in the join each triangle is produced six times.  It is not hard to eliminate five of the copies; just produce $(X,Y,Z)$ as an output if and only if $X<Y<Z$ according to a chosen ordering of the nodes.  However, this approach is somewhat like counting cows in a field by counting the legs and dividing by 4.  That's not too bad, but when we count subgraphs with larger numbers of nodes, we wind up counting centipedes or millipedes that way, and the idea cannot be sustained.

Thus,  we shall assume an ordering $(<)$ of the nodes, and the tuple $E(a,b)$ will be in relation $E$ if and only if $(a,b)$ is an edge of the graph and $a<b$.  In this case, each triangle is discovered exactly once.

Following the method of \cite{AU10}, to compute the join
$$E(X,Y)\bowtie E(Y,Z)\bowtie E(X,Z)$$
we must by symmetry hash each of the variables $X$, $Y$, and $Z$ to the same number of buckets $b$.  An ordered triple of buckets identifies a reducer.  If we hash each variable to $b$ buckets using hash function $h$, then there are $b^3$ reducers.  If $E(u,v))$ is a tuple of $E$, this edge is sent by its mapper to $3b-2$ reducers in three groups:

\begin{enumerate}

\item
Treated as an edge $E(X,Y)$, it is sent to the $b$ different reducers whose triple is $[h(u),h(v),z]$ for any $z=1,2,\ldots,b$.

\item
Treated as an edge $E(Y,Z)$, it is sent also to the $b$ reducers $[x,h(u),h(v)]$ for any $x$.

\item
Treated as $E(X,Z)$, it is sent to the reducers $[h(u),y,h(v)]$ for any $y$.

\end{enumerate}
However, it is easy to see that regardless of whether or not $h(u)=h(v)$, exactly two of these reducers will be the same.  There are two cases:

\begin{itemize}

\item[a)]
If $h(u)\ne h(v)$, then no reducer in the first group can equal a reducer in the second group, because their middle components must be different.  However, the reducer of the third group, with $y=h(v)$ will be in the first group and the reducer with $y=h(u)$ will be in the second group.

\item[b)]
If $h(u)=h(v)$, then the reducer of the first group with $z=h(u)$, the reducer of the second group with $x=h(u)$, and the reducer of the third group with $y=h(u)$ are all the same reducer, but no other reducers are the same.

\end{itemize}
Thus the communication cost for this algorithm is $m(3b-2)$.\footnote{In practice, it is unlikely we would try to take advantage of the duplication of two reducers, but would in fact create $3b$ key-value pairs for each edge.  The redundancy is small for large $b$, and the map-reduce environment would make it tricky to avoid the redundancy.}

Each of the $b^3$ reducers computes the join for the tuples it is given.
The triangle consisting of nodes $u$, $v$, and $w$, where $u<v<w$ is discovered only by the reducer $[h(u),h(v),h(w)]$.  That is, we may substitute $u$ for $X$, $v$ for $Y$, and $w$ for $Z$, and the tuples $E(u,v)$, $E(v,w)$, and $E(u,w)$ surely exist.  However, if we make any other substitution of $u$, $v$, and $w$ for $X$, $Y$, and $Z$, at least one pair of variables will be out of order and the corresponding tuple will not exist in $E$, although its reverse does.

\subsection{Ordering Nodes by Bucket}
\label{improved-subsect}

We can improve the algorithm of Section~\ref{mj-subsect} by exploiting the fact that the ordering of nodes is subject to our choice and thus can be related to the bucket numbers.  Let $h$ be a hash function that maps nodes to $b$ buckets.  When ordering nodes, think of node $u$ as a pair consisting of $h(u)$ followed by $u$ itself.  That is, all nodes of bucket 1 precede all nodes of bucket 2, which precede nodes of bucket 3, and so on.  Within a bucket, the name of the node breaks ties.

The advantage to this approach is that many of the lists of three buckets now correspond to reducers that get no triangles, and therefore we do not need reducers for these lists.\footnote{Or, since it is likely that the number of reducers will be chosen first, our reasoning allows us to use less communication for a fixed number of reducers, as we shall see.}
We only need a reducer for a list $[i,j,k]$ if $1\le i \le j \le k \le b$.  How many such lists are there?  It is the same as the number of strings with $b-1$ 0's and three 1's, that is, $\binom{b+2}{3} = (b+2)(b+1)b/6$.  In proof, we show that there is a 1-1 correspondence between the lists and the strings just described.  Consider a string of 0's and 1's where the first 1 is in position $p$, the second 1 is in position $q$, and the third in position $r$.  This string corresponds to $[i,j,k]$ where $i = p$, $j = q-1$, and $k = r-2$.  Each list corresponds to a unique string, and each string corresponds to a unique list; we prove a more general observation in Theorem~\ref{order-buckets-th}.  Thus, like the Partition Algorithm, the method of this section uses approximately $b^3/6$ reducers.

As with the algorithm of Section~\ref{mj-subsect}, each reducer handles the portion of the data graph that it is given.  A triangle is discovered by only one reducer~-- the reducer that corresponds to the buckets of its three nodes, in sorted order.

We claim that the communication cost for this algorithm is $b$ per edge.
Let $(u,v)$ be an edge of the graph.  This edge must be sent to all and only the reducers corresponding to the sorted list consisting of $h(u)$, $h(v)$, and any one of the buckets from 1 to $b$.  Note that some of these lists have repeating bucket numbers, as must be the case since some triangles have two or three nodes that hash to the same bucket.

The argument that \cite{SV11} used to show that the map-reduce implementation of Partition has the same order of computation time as the serial algorithm also works for this algorithm.
The serial algorithm takes $O(m^{3/2})$ time on a graph of $m$ nodes.  If we hash nodes to $b$ buckets, each of the $m$ edges goes to $b$ reducers.  There are $O(b^3)$ reducers, so each reducer gets an expected $O(m/b^2)$ edges.  The time this reducer will take to find triangles in its portion of the graph is $O\bigl((m/b^2)^{3/2}\bigr) = O(m^{3/2}/b^3)$.  But since there are $O(b^3)$ reducers, the total computation cost of all the reducers is $O(m^{3/2})$, exactly as for the serial algorithm.

\subsection{Comparison of Triangle-Finding Algorithms}
\label{tri-compare-subsect}

The three algorithms are rather similar in critical measures, but the one given in Section~\ref{improved-subsect} is best for communication cost by a small amount.  First, let us assume that there are $k$ reducers, and $k$ is large enough that $b$ plus or minus a constant can be approximated by $b$.  Then Fig.~\ref{asymp-fig} gives the communication cost for each of the algorithms.
Using the same number of reducers, the algorithm of Section~\ref{improved-subsect} beats the Partition Algorithm for communication cost by a factor of $3/2$ and beats the algorithm of Section~\ref{mj-subsect} by a factor of $3/\sqrt[3]{6} = 1.65$.
The computation costs for the three algorithms are similar, but only the Partition Algorithm finds some triangles more than once, requiring extra time to compensate for that effect.\footnote{While it is true that the difference in communication cost is only a constant factor, note that there is no ``big-oh'' involved.  For each of the algorithms, the key-value pairs are the same; they consist of a list of three bucket numbers, an edge, and an indication of between which of the three pairs of nodes the edge lies.}

\begin{figure}[htfb]

\begin{center}
\begin{tabular}{| l || c | c|}
\hline
Algorithm & Buckets $b$ & Communication Cost\\
\hline\hline
Partition & $\sqrt[3]{6k}$ & $3m\sqrt[3]{6k}/2$\\
\hline
Section~\ref{mj-subsect} & $\sqrt[3]{k}$ & $3m\sqrt[3]{k}$\\
\hline
Section~\ref{improved-subsect} & $\sqrt[3]{6k}$ & $m\sqrt[3]{6k}$\\
\hline

\end{tabular}
\end{center}

\caption{Asymptotic performance of three triangle-finding algorithms}
\label{asymp-fig}

\end{figure}

For a comparison using specific values of $b$, note that $216 = 6^3$ and $220 = \binom{12}{3}$.   Figure~\ref{216-220-fig} makes the comparison, using 216 reducers for the algorithm of Section~\ref{mj-subsect} and (almost the same number) 220 reducers for the other two algorithms.  We see that the asymptotic comparison holds up for these reasonable numbers of reducers.
\begin{figure}[htfb]

\begin{center}
\begin{tabular}{| l || c | c| c |}
\hline
Algorithm &Buckets $b$ &  Reducers & Communication Cost\\
\hline\hline
Partition & 12 & 220 & $13.75m$\\
\hline
Section~\ref{mj-subsect} & 6 & 216&  $16m$\\
\hline
Section~\ref{improved-subsect} & 10 &  220 & $10m$\\
\hline

\end{tabular}
\end{center}

\caption{Comparison of algorithms for specific numbers of reducers.  For Partition, the 12 ``buckets'' is really the number of sets into which the nodes are partitioned.  The constant 13.75 in the communication cost for this algorithm is $\frac32(b-1)(b-2)/b$ for $b=12$.
The constant 16  for the algorithm of Section~\ref{mj-subsect} is $3b-2$ for $b=6$}
\label{216-220-fig}

\end{figure}
%\newpage

\section{Sample Graphs and Conjunctive Queries}
\label{sg-cq-sect}

When we consider finding instances of sample graphs more complex than triangles, the multiway-join approach continues to apply and to be superior to a cascade of two-way joins. However, the joins must be constrained by arithmetic comparisons among nodes, to enforce certain node orders.  Those constraints in turn are needed to avoid producing an instance of a sample graph more than once.  A natural notation for such joins-plus-selections is conjunctive queries (abbreviated CQ) with arithmetic comparisons (\cite{YuOzs} \cite{Gupta}), and we shall use this notation in what follows.

Suppose we are searching for instances of a sample graph $S$.  A CQ will have a variable corresponding to each node of $S$, and it will have a relational subgoal $E(X,Y)$ whenever $S$ has an edge between nodes $X$ and $Y$.
Relation $E(X,Y)$ contains each edge of the data graph exactly once, and does so in the order $X<Y$; i.e., the node in the first argument precedes the node in the second argument according to some given order of the nodes.  For the case of a general sample graph $S$, we construct CQ's for $S$  by a three-step process:

\begin{enumerate}
\item Typically, the sample graph $S$ will have a nontrivial automorphism group.\footnote{An {\em automorphism} is a 1-1 mapping from nodes to nodes that preserves the presence of an edge.} Thus, some orders of the nodes of $S$ are automorphic to others, and the node orders fall into equivalence classes.  Select one representative from each equivalence class.

\item For each chosen order of the nodes of $S$, write a CQ that uses subgoals $E$ with arguments chosen to respect that ordering.  The arithmetic condition for the CQ enforces the ordering.  That is, we believe that $E$ will contain only edges oriented in the direction given by the ordering $<$, but we do not rely on that assumption when expressing CQ's.

\item For most sample graphs $S$, there will be several selected CQ's that have the same orientation of all the edges of $S$; i.e., the relational subgoals of the CQ's will be identical, although the arithmetic conditions will differ. Combine CQ's with identical edge orientations by taking the logical OR of the arithmetic conditions.
\end{enumerate}

\noindent{\em Remark.}
Regarding step~(1) and the size of the automorphism group of $S$, we note that using such methods to mine data, e.g., for discovering potential threats or for providing recommendations, mostly involves answering questions with a lot of symmetry in them (see e.g. that mentioned in Section~\ref{apps-subsect}). Hence, we expect that $S$ will typically have a relatively large automorphism group, which is exploited by our approach.
On the other hand, almost all very large random graphs are asymmetric, i.e., they have no nontrivial automorphisms. However, the sample graphs are typically very small, and most small graphs have nontrivial automorphism groups. For example, all graphs with 5 nodes have nontrivial automorphisms, and there are only 4 asymmetric graphs with 6 nodes (see e.g., \cite{ER63}).

\subsection{Generating CQ's from Orderings}
\label{cq-from-order-subsect}

We begin with step~(2), the simplest point.  Given an ordering $X_1,X_2,\ldots,X_p$ for the nodes of the sample graph $S$, the CQ has:

\begin{enumerate}
\item A relational subgoal $E(X_i,X_j)$ if $S$ has an edge $(X_i,X_j)$ and $i<j$.

\item An arithmetic subgoal $X_i<X_{i+1}$ for all $i=1,2,\ldots,p-1$.
\end{enumerate}

\begin{figure}[htfb]
\centerline{\includegraphics[width=0.22\textwidth]{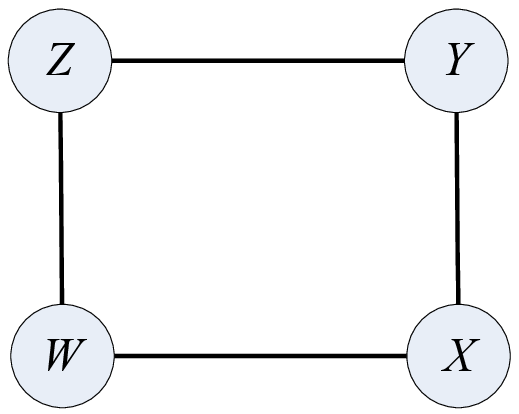}}
\caption{The square}
\label{sq-fig}
\end{figure}

\begin{example}
\label{sq-ex}
Let us consider the square as our sample graph, with nodes labeled by variables as in Fig.~\ref{sq-fig}.  There are 24 orders for the four variables $W$, $X$, $Y$, and $Z$ that label the nodes.  However,  as we shall see, three CQ's suffice to get all squares.  Consider the order $W<X<Y<Z$.  The body of the CQ corresponding to this order is

\begin{verbatim}
     E(W,X) & E(X,Y) & E(Y,Z) & E(W,Z) &
         W<X & X<Y & Y<Z
\end{verbatim}
Notice that each of the four edges is represented by a subgoal $E$ with the two arguments in the required order.
\end{example}

\subsection{Exploiting Automorphisms}
\label{cq-perm-subsect}

Suppose we wish to find all instances of a $p$-node sample graph $S$ in data graphs.  Then for each of the $p!$ orders of the nodes of $S$ there is a CQ.  This CQ has an $E$ subgoal for each edge, with its two arguments in the order required.  The arithmetic condition is that the variables are in the given order.
In Example~\ref{sq-ex} we gave one such CQ when $S$ is the square.  There are 23 others.

However, often $S$ will have a nontrivial automorphism group.  If so, it is not necessary to use a CQ for each permutation.  Rather, one CQ per element of the quotient group suffices.

\begin{theorem}\label{automorphism-th}
Let $S$ be a sample graph with $p$ nodes. We can discover exactly once every instance of $S$ in a data graph $G$ by applying to $G$ one CQ for each member of the group that is the quotient of the symmetric group of $p$ elements (permutations of $p$ things) with the automorphism group of $S$.
\end{theorem}

\begin{proof}
Suppose graph $G$ has an instance of $S$, say $G_0$, and $\mu$ is an automorphism on $S$.  Let $\nu$ map the nodes of $G_0$ to the nodes of $S$.  Then $\mu\circ\nu$ is also a mapping from $G_0$ to $S$.  The nodes of $G_0$ are in some order.  When we apply $\nu$ to these nodes, they induce a particular order on the nodes of $S$.   That order gives rise to a CQ $Q_1$, and we know that the nodes of $G_0$ will satisfy $Q_1$.  But we can also map $G_0$ to $S$ using the mapping $\mu\circ\nu$, and this mapping induces another order on the nodes of $S$, an ordering that has CQ $Q_2$.  It follows that $G_0$ will also satisfy $Q_2$.  Since the inverse of an automorphism is also an automorphism, the same argument shows that every instance of $S$ identified by $Q_2$ is also identified by $Q_1$.

Therefore, we neither want nor need to use both $Q_1$ and $Q_2$ when searching for instances of $S$.  In general, if $S$ has $p$ nodes, we start with the symmetric group of $p$ elements and take the quotient of that group with the automorphism group of $S$.  The quotient group consists of classes of orders of the nodes of $S$.  We choose one representative ordering from each class.  There can be no automorphisms between representatives of different classes, so the CQ's for these representatives can never produce the same instance of $S$ within $G$.  On the other hand, every instance of $S$ in $G$ has some ordering of its nodes, and therefore is discovered by the CQ for that ordering.  Therefore, it is also discovered by the CQ representing the class to which that ordering belongs in the quotient group.
\end{proof}

\begin{example}\label{sq-ordering-ex}
The square has an automorphism group of size eight and a symmetric group of size 24.
The automorphisms of the square are described by allowing a rotation of the square to any of four positions.  Additionally, we can choose to ``flip'' the square (turn the paper on which it is written over) or not.  For example, using the node names of Fig.~\ref{sq-fig} the automorphisms of the order $WXYZ$ are the identity, $XYZW$ (rotate 90 degrees clockwise), $YZWX$ (rotate 180 degrees), $ZWXY$ (rotate 270 degrees), and the four flips of these rotations: $WZYX$, $ZYXW$, $YXWZ$, and $XWZY$.  Intuitively, these orders are those in which the four nodes of the square form an increasing sequence in one direction around the square with any node as the starting point.

Since $24/8=3$, we expect there are two other sets of orders that are automorphic.  One of these is the orders in which two opposite corners are each higher than the other two opposite corners.  These orders are $WYXZ$, $YWXZ$, $WYZX$, $YWZX$, $XZWY$, $ZXWY$, $XZYW$, and $ZXYW$.  The other group covers the cases where two opposite corners are the extreme values (low and high), and the other two nodes are in the middle.  These are $WXZY$, $WZXY$, $YXZW$, $YZXW$, $XWYZ$, $XYWZ$, $ZWYX$, and $ZYWX$.

Pick representatives, say $WXYZ$, $WYXZ$, and $WXZY$, for each of the three groups.  Then the three CQ's that together find each square exactly once are:

\begin{verbatim}
     E(W,X) & E(X,Y) & E(Y,Z) & E(W,Z) &
         W<X & X<Y & Y<Z
     E(W,X) & E(Y,X) & E(Y,Z) & E(W,Z) &
         W<Y & Y<X & X<Z
     E(W,X) & E(X,Y) & E(Z,Y) & E(W,Z) &
         W<X & X<Z & Z<Y
\end{verbatim}
Notice that all three have the subgoals $E(W,X)$ and $E(W,Z)$, but differ in the orders of the arguments of the second and third subgoals.  Also, in the second and third CQ's, the arithmetic condition enforces a total order that is stronger than what the order of arguments of $E$ implies.
\end{example}

\subsection{Exploiting Edge Orientations}
\label{edge-orient-subsect}

For some sample graphs $S$, the CQ's generated by the method of Section~\ref{cq-perm-subsect} will repeat some edge orientations.  If so, these CQ's can be combined if we replace the arithmetic conditions from each by the OR of those conditions.\footnote{In some cases, the OR of arithmetic conditions cannot be expressed in the form needed for a conjunctive query.  However, since we implement each CQ as a multiway join followed by a selection, and any selection condition, whether or not it is the AND of simple comparisons, can be implemented at the end of the Reduce function, we need not worry about the nature of the selection condition in what follows.}
In Example~\ref{sq-ordering-ex} there were only three CQ's for the square, each with a different edge orientation.

\begin{figure}[htfb]
\centerline{\includegraphics[width=0.37\textwidth]{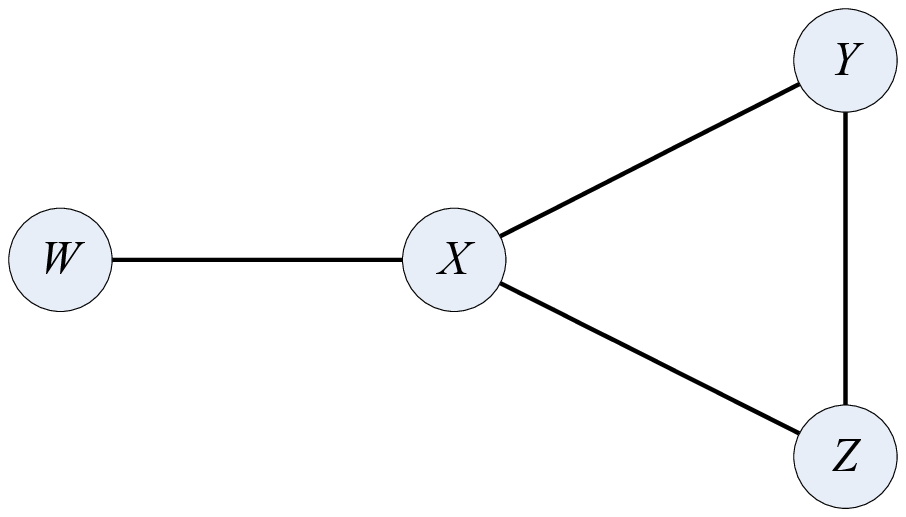}}
\caption{The lollipop}
\label{lol-fig}
\end{figure}

However, for other sample graphs, such as the ``lollipop,'' shown with names/variables for its nodes in Fig.~\ref{lol-fig}, edge orientation allows significant simplification.
Since the lollipop has four nodes, there are 24 orders, but its automorphism group has only two members: the identity and the mapping that swaps $Y$ with $Z$.  Thus, the quotient group has twelve members.  We can break the symmetry of the automorphisms by requiring $Y<Z$.  That inequality restricts an edge, so we would expect that there are eight orientations of the edges.  However, two of these orientations are impossible, since we cannot have both $Z<X$ and $X<Y$, or there would be a contradiction with $Y<Z$.  Thus, we expect only six CQ's suffice.

\begin{figure}[htfb]
{%\small
\begin{center}
\begin{tabular}{r l l}
 & Order & Conjunctive Query (Relational Subgoals Only)\\
\hline
1. & $W<X<Y<Z$ & $E(W,X)$ \& $E(X,Y)$ \& $E(X,Z)$ \& $E(Y,Z)$\\
2. & $W<Y<X<Z$ & $E(W,X)$ \& $E(Y,X)$ \& $E(X,Z)$ \& $E(Y,Z)$\\
3. & $W<Y<Z<X$ & $E(W,X)$ \& $E(Y,X)$ \& $E(Z,X)$ \& $E(Y,Z)$\\
4. & $X<W<Y<Z$ & $E(X,W)$ \& $E(X,Y)$ \& $E(X,Z)$ \& $E(Y,Z)$\\
5. & $Y<W<X<Z$ & $E(W,X)$ \& $E(Y,X)$ \& $E(X,Z)$ \& $E(Y,Z)$\\
6. & $Y<W<Z<X$ & $E(W,X)$ \& $E(Y,X)$ \& $E(Z,X)$ \& $E(Y,Z)$\\
7. & $X<Y<W<Z$ & $E(X,W)$ \& $E(X,Y)$ \& $E(X,Z)$ \& $E(Y,Z)$\\
8. & $Y<X<W<Z$ & $E(X,W)$ \& $E(Y,X)$ \& $E(X,Z)$ \& $E(Y,Z)$\\
9. & $Y<Z<W<X$ & $E(W,X)$ \& $E(Y,X)$ \& $E(Z,X)$ \& $E(Y,Z)$\\
10. & $X<Y<Z<W$ & $E(X,W)$ \& $E(X,Y)$ \& $E(X,Z)$ \& $E(Y,Z)$\\
11. & $Y<X<Z<W$ & $E(X,W)$ \& $E(Y,X)$ \& $E(X,Z)$ \& $E(Y,Z)$\\
12. & $Y<Z<X<W$ & $E(X,W)$ \& $E(Y,X)$ \& $E(Z,X)$ \& $E(Y,Z)$\\
\end{tabular}
\end{center}
}
\caption{Twelve CQ's for the lollipop, omitting the arithmetic comparisons}
\label{lol-cqs-fig}

\end{figure}

\begin{example}
\label{lol-ex}
In Fig.~\ref{lol-cqs-fig} we see the twelve CQ's that come from the twelve orders with $Y<Z$.
The three arithmetic subgoals that enforce the order for each CQ are omitted to save space.
Observe that all twelve have subgoal $E(Y,Z)$, as they must.  However, the twelve divide into six groups with identical relational subgoals.  These groups are summarized in Fig.~\ref{lol-groups-fig}.  The orientations are represented by listing the low end of each edge first.  For instance the first group corresponds to the edge orientation where $W<X$, $X<Y$, and $X<Z$.

\begin{figure}[htfb]

\begin{center}
\begin{tabular}{l l}
Orientation & CQ's\\
\hline
$WX,XY,XZ$ & 1\\
$WX,YX,XZ$ & 2, 5\\
$WX,YX,ZY$ & 3, 6, 9\\
$XW,XY,XZ$ & 4, 7, 10\\
$XW,YX,XZ$ & 8, 11\\
$XW,YX,ZY$ & 12\\
\end{tabular}
\end{center}

\caption{Grouping CQ's for the lollipop by edge orientation}
\label{lol-groups-fig}

\end{figure}

The first and last groups have only one CQ, so the first and twelfth CQ's are carried over intact.
The second group consists of CQ's (2) and (5).  Notice that their arithmetic conditions differ only in that (2) has $W<Y$ while (5) has $Y<W$.  The logical OR of the conditions thus replaces these two inequalities by $W\ne Y$.   Similarly , the fifth group $\{8,11\}$ has conditions that differ only in the order of $W$ and $Z$, so we replace the two inequalities on $W$ and $Z$ by $W\ne Z$ to obtain the logical OR.

Now consider the third group $\{3,6,9\}$.  The three orders of variables for these CQ's all have $Y<Z<X$ and $W<X$.  However, $W$ can appear anywhere in relation to $Y$ and $Z$.  The OR of the three conditions is thus $Y<Z$, $Z<X$, $W<X$, $W\ne Y$, and $W\ne Z$.  The fourth group, $\{4,7,10\}$ is handled similarly, except in this group $X$ is lowest rather than the highest of the variables.  Figure~\ref{lol-oriented-fig} shows the six resulting CQ's for the lollipop sample graph.
\end{example}

\begin{figure}[htfb]
{%\small
\begin{center}
\begin{tabular}{l}
$E(W,X)$ \& $E(X,Y)$ \& $E(X,Z)$ \& $E(Y,Z)$ \&\\
~~~~$W<X$ \& $X<Y$ \& $Y<Z$\\
$E(W,X)$ \& $E(Y,X)$ \& $E(X,Z)$ \& $E(Y,Z)$ \&\\
~~~~$W\ne Y$ \& $Y<X$ \& $X<Z$\\
$E(W,X)$ \& $E(Y,X)$ \& $E(Z,X)$ \& $E(Y,Z)$ \&\\
~~~~$W<X$ \& $Y<Z$ \& $Z<X$ \& $W\ne Y$ \& $W\ne Z$\\
$E(X,W)$ \& $E(X,Y)$ \& $E(X,Z)$ \& $E(Y,Z)$ \&\\
~~~~$X<W$ \& $X<Y$ \& $Y<Z$ \& $W\ne Y$ \& $W\ne Z$\\
$E(X,W)$ \& $E(Y,X)$ \& $E(X,Z)$ \& $E(Y,Z)$ \&\\
~~~~$Y<X$ \& $X<W$ \& $W\ne Z$\\
$E(X,W)$ \& $E(Y,X)$ \& $E(Z,X)$ \& $E(Y,Z)$ \&\\
~~~~$Y<Z$ \& $Z<X$ \& $X<W$\\
\end{tabular}
\end{center}
}
\caption{Six CQ's for the lollipop after combining CQ's with the same orientation}
\label{lol-oriented-fig}
\end{figure}

\section{Evaluation of CQ's with Optimal Communication Cost}
\label{opt-eval-sect}

We can apply the method of \cite{AU10} to evaluate each of the CQ's for a sample graph $S$ optimally as regards the communication cost.  There are three broad approaches to doing so:

\begin{enumerate}
\item {\em CQ-Oriented Processing}. Perform a separate join for each CQ.  This approach never dominates the others, but we shall begin our discussion of evaluation by focusing on a single CQ in Section~\ref{1-cq-opt-subsect}.

\item {\em Variable-Oriented Processing}. Treat all the CQ's as if they were a single join of the relations for the edges of $S$.  More precisely, if the subgoal $E(X,Y)$ appears in each CQ for $S$, then the relation for the edge $(X,Y)$ is $E$.  However, if both $E(X,Y)$ and $E(Y,X)$ appear among different CQ's for $S$, then the relation for the edge $(X,Y)$ is two copies of $E$, one with the attributes in the order $(X,Y)$ and the other in the order $(Y,X)$.  In that case, the relation for the edge $(X,Y)$ is twice as large as it would be if the edge appeared in only one orientation among all the CQ's.  Note that the reducers still evaluate each of the CQ's separately, although there might be some common subexpressions that can simplify the work.

\item {\em Bucket-Oriented Processing}. Here, we use the same number of buckets for each of the variables in each CQ. For each nondecreasing sequence of bucket numbers, evaluate each CQ using the edges whose ends are in buckets that appear in the sequence.
\end{enumerate}
One might suppose that there is a fourth approach, where we treat $E$ as a relation of undirected edges and include both $E(a,b)$ and $E(b,a)$.  We then take a single multiway join and eliminate duplicate copies of instances of $S$ by enforcing some constraints on the order of nodes in the instance.  However, it is easy to see that this approach is never superior to the variable-oriented method and can be worse.

%We shall begin by applying the first approach and defer the second to Section~\ref{combine-cq-subsect}.

\subsection{Optimization of Single CQ's}
\label{1-cq-opt-subsect}

To review \cite{AU10} and \cite{AU11-Datalog2.0}, the way to optimize the map-reduce evaluation of a CQ is:

\begin{itemize}

\item
For each variable of the CQ $X$, there is a {\em share} $x$ that is the number of buckets into which values of $X$ are hashed.

\item
Each reducer is identified by a list of bucket numbers, one for each variable of the CQ, in a fixed order.

\item
The communication cost for evaluating the CQ is a sum of terms, one for each relational subgoal of the CQ.  This term is the product of the size of the relation for that subgoal and all the shares of the variables that {\em do not} appear in that subgoal.

\item
The minimum value of this sum occurs when the sums of certain subsets are all equal.  There is one subset for each share, and that subset consists of all the terms in which that share appears.

\end{itemize}

\begin{example}
\label{lol-opt-ex}
The first of the six CQ's in Fig.~\ref{lol-oriented-fig} is

\begin{verbatim}
     E(W,X) & E(X,Y) & E(X,Z) & E(Y,Z)
\end{verbatim}
We have dropped the arithmetic comparisons.  They will be implemented by a selection after performing the join, at the same reducers that produce the join, so they have no effect on the communication cost.  There are four shares $w$, $x$, $y$, and $z$, corresponding to the four variables of the CQ.  However, a theorem of \cite{AU10} says that when one variable is {\em dominated} by another (the first variable appears only in terms where the second appears, then the share of the first may be taken as 1; i.e., the dominated variable may be ignored when determining the reducers to which a tuple is sent.  As $W$ appears only where $X$ appears, we shall assume $w=1$ and drop share $w$ from formulas.

All four terms have the same relation $E$, so we shall use $e$ as the size of $E$.  Thus, all terms in the expression for communication cost will have $e$ as a factor.  The terms for our example CQ are thus:
$$eyz + ez + ey + ex$$
In explanation, the first term $eyz$ comes from subgoal $E(W,X)$.  It consists of the size $e$ times the shares of the variables $Y$ and $Z$ that do not appear in the subgoal.  The second term $ez$ comes from the subgoal $E(X,Y)$.  The missing variables are $W$ and $Z$, but recall that the dominated $W$ has $w=1$ so there is no factor $w$ needed.  The last two terms are derived from the last two subgoals in a similar manner.

Now, we must derive the subsets of terms that are required to be equal at the optimum point.  The share $x$ is present only in the last term, so that term is one of the sums.  Share $y$ is present in the first and third terms, so another subset is $eyz+ey$.  Share $z$ appears in the first two terms, so the last subset is $eyz+ez$.  The minimum communication cost thus occurs when
$$ex = eyz+ey = eyz+ez$$

From the above equalities, we can deduce $z=y$ and $x=y^2+y$.  Since the number of reducers is the product of all the shares, we now can pick a value of $y$ and know completely how to replicate edges to evaluate the CQ.  For instance, pick $y=5$.  Then $x=30$, $z=5$, and there are $xyz=750$ reducers.  Each edge is replicated as a tuple for the first subgoal $E(W,X)$ to $yz=25$ reducers.  It is replicated 5 times as a tuple for the each of the second and third subgoals, and it is replicated 30 times as a tuple for the last subgoal, a total of 65 times.
\end{example}

\subsection{Regular Sample Graphs}
\label{eq-deg-subsect}

For regular sample graphs, the optimum way to assign shares is much simpler than what we saw in Example~\ref{lol-opt-ex}.

\begin{theorem}
\label{shares-th}
If subgaph $S$ has all nodes of degree $d$, there are $p$ nodes, and we wish to use $k$ reducers in the map-reduce evaluation of one of the CQ's for $S$, then each node gets a share $\sqrt[p]{k}$ in the optimum assignment of shares.
\end{theorem}

\begin{proof}
If $e$ is the number of edges in the graph to which the CQ is applied, then the expression for communication cost consists of $pd/2$ terms, one for each edge of $S$.  The term for an edge is the product of $e$ and the shares for the $p-2$ nodes that are not ends of the edge.
Each node appears in $\frac12 pd-d$ of these terms, so the conditions for optimality are the equalities among $p$ expressions, each of which is the sum of $\frac12 pd-d$ terms.  Moreover, each of these terms is the product of $e$ and $p-2$ different shares.
Thus, all equalities are satisfied when all the shares are the same.
Since the product of the shares must be $k$, the number of reducers, it follows that the minimum communication cost is obtained when all shares are $\sqrt[p]{k}$.
\end{proof}

Theorem~\ref{shares-th} applies to many interesting sample graphs,
including all cycles, all complete graphs, and hypercubes of any dimension.
An important consequence of this theorem is that the hashed-based ordering of nodes that we used in Section~\ref{improved-subsect} for triangles can be exploited for any regular sample graph.
If the sample graph has $p$ nodes, then many of the reducers will get no instances of the sample graph and need not be executed.  That effect in turn lowers the replication of edges significantly.

\begin{theorem}
\label{order-buckets-th}
Suppose $S$ is a sample graph of $p$ nodes, and $Q$ is a CQ that generates instances of $S$ that have a particular order for the nodes of $S$.
If each variable of $Q$ has share $b$, the same hash function $h$ is used for each variable, and the order of nodes is determined first by its hash value and then by the identifier of the node to break ties, then the number of reducers that need to be executed is $\binom{b+p-1}{p}$.
\end{theorem}

\begin{proof}
As in Section~\ref{improved-subsect}, we can count the number of useful reducers by comparing them to certain binary strings.  First, observe that since $Q$ generates only instances of $S$ in which a particular order of the variables holds, the buckets for those variables must form a nondecreasing sequence.  Wlog we can assume the order of the variables is
$X_1<X_2<\cdots<X_p$
and that the identifier for a reducer is the list
$[h(X_1),h(X_2),\ldots,h(X_p)]$.

Since $h(X_1)\le h(X_2)\le\cdots\le h(X_p)$, the number of useful reducers is equal to the number of sequences of integers
$$1\le i_1\le i_2\le\cdots\le i_p\le b$$
These sequences are in 1-1 correspondence with the number of strings of 0's and 1's with $p-1$ 0's and $b$ 1's.  Specifically, we can identify a sequence of integers $1\le i_1\le i_2\le\cdots\le i_p\le b$ with the string where the $j$th 1 is in position $i_j+j-1$.

We must prove that no two sequences yield the same string and no two strings correspond to the same sequences.  Suppose two sequences differ first in position $i$.  Then the corresponding strings must have their $i$th 1's in different positions and therefore are different strings.  Conversely, suppose two strings agree in the positions of their first $i-1$ 1's, but disagree on the positions of their $i$th 1's.  Then the corresponding sequences differ in their $i$th components.  Finally, the maximum position that can hold a 1 is when $i_p=b$, in which case the position is $p+b-1$.  Thus, the correspondence is 1-1.  Since there are $\binom{p+b-1}{p}$ strings, that is also the number of reducers needed.
\end{proof}

\subsection{Variable-Oriented Processing}
\label{combine-cq-subsect}

%There is an option to evaluate several CQ's using the same set of reducers.  When we do, we lose the option to apply Theorem~\ref{order-buckets-th}, because there is no particular order to the variables of the CQ's.\footnote{It is possible to extend this theorem to the case where different variables hash to different numbers of buckets, but the count of useful reducers becomes very hard to express in general, since there can be partial overlaps among the nodes in the buckets for different variables.}
%Moreover, since some of the edges may appear among the CQ's in both orientations, the subgoals may have effective relation sizes that are not uniform, even though they are each based on the same set of undirected edges.
%

In this approach, we treat all the CQ's as if they were one.  That is, there is a reducer for each list of buckets, one for each variable.  The number of buckets for different variables may differ.

Each CQ for a sample graph $S$ has one subgoal for each undirected edge.  It may be that for some of these edges, the orientation is the same in each CQ.  For these subgoals, each edge is communicated from mappers to reducers in only one orientation, while for the other subgoals, each edge must be communicated in both orientations.  The effect is that if $e$ is the number of undirected edges, the relation size for a subgoal whose edge appears in both orientations among the CQ's is $2e$ rather than $e$.
This change makes it somewhat harder to find the optimum shares, and in particular, Theorem~\ref{order-buckets-th} cannot be applied.  However, the minimum communication cost can still be obtained with care.

\begin{example}
\label{sq-cc-ex}
Consider the three CQ's (with arithmetic subgoals omitted)

\begin{verbatim}
     E(W,X) & E(X,Y) & E(Y,Z) & E(W,Z)
     E(W,X) & E(Y,X) & E(Y,Z) & E(W,Z)
     E(W,X) & E(X,Y) & E(Z,Y) & E(W,Z)
\end{verbatim}
derived for the square in Example~\ref{sq-ordering-ex}.
In these CQ's, the edges $(W,X)$ and $(W,Z)$ appear in only one orientation each, while the other two edges appear in both orientations.  The expression for the communication cost is thus
$$eyz + 2ewz + 2ewx + exy$$
where we conventionally use the corresponding lower-case letter to represent the share for a variable of the CQ.
To solve for the shares, we must satisfy the equalities
$$2ewz+2ewx = 2ewx+exy = eyz+exy = eyz+2ewz$$
Interestingly, these equations do not provide unique values for the shares, even under the constraint that $wxyz=k$, where $k$ is the desired number of reducers.  However, we can derive the simple equalities $x=z$ and $y=2w$.  We are free to select values for the shares within these constraints; any choice will provide the same, optimum communication cost, even though the shares themselves differ.  For instance, a simple choice would be $x=z=1$, $w=\sqrt{k/2}$, and $y=\sqrt{2k}$.  With that choice (or any other choice that satisfies the needed $x=z$, $y=2w$, and $wxyz=k$), the communication cost per edge is $4\sqrt{2k}$.
\end{example}

\subsubsection{Variable-Oriented Processing for Regular Sample Graphs}
\label{combine-app-regular}

Next, we show how the above approach applies to special cases of
regular graphs. For two families of regular sample graphs, the idea illustrated in Example~\ref{sq-cc-ex} can be generalized.

\begin{theorem}\label{shares2-th}
Let $S$ be regular and have $p$ nodes.  Each edge of $S$ is either {\em bidirectional} (it appears in both directions in some CQ for $S$) or {\em unidirectional} (it appears in only one direction among these CQ's).
Suppose it is possible to partition the nodes of $S$ into two sets $S_1$ and $S_2$ such that one of following two cases holds:

\begin{itemize}
\item[(a)]
Every bidirectional edge runs between nodes of $S_1$ and the
unidirectional edges run between an $S_1$ and an $S_2$ node, or

\item[(b)]
Every bidirectional edge runs between an $S_1$ and an $S_2$ node, and
every unidirectional edge runs between $S_2$ nodes.
\end{itemize}

Then in either case, the shares for the nodes in $S_1$ are all equal, the shares of the nodes in $S_2$ are all equal, and the shares of the nodes in $S_1$ are twice the shares of the nodes in $S_2$.
\end{theorem}

\begin{proof}
The proof is almost the same as that for Theorem~\ref{shares-th}, with the exception that the terms for the bidirectional edges have an additional factor 2, because their relations are twice the size of the relations for the unidirectional edges.  If we divide all terms by $e$ (the size of the edge relation) and also divide by $k$, the product of all the shares of all the nodes, then the term for an edge between nodes $A$ and $B$ becomes either $1/ab$ or $2/ab$, where $a$ and $b$ are the shares for nodes $A$ and $B$, respectively.  The numerator is 1 if the edge is unidirectional and 2 if it is bidirectional.  It is easy to check that in both cases the product $ab$ is twice as large if edge $(A,B)$ is bidirectional than it is if $(A,B)$ is unidirectional.
Thus, all terms contribute the same to the sum for each edge.
Since $S$ is regular, each sum has the same number of terms so all these sums are equal.
\end{proof}

\begin{example}
\label{shares2-ex}
Consider the cycle $C_p$ with nodes
$$X_1,X_2,\ldots,X_p$$
in that order (as in Fig.~\ref{cycle-fig}).
Assume CQ's are selected in the standard way, so $X_1<X_2$, $X_1<X_p$, and $X_2<X_p$ to break the automorphisms.  As a result, the only unidirectional edges are $(X_1,X_2)$ and $(X_1,X_p)$.  Let $S_2 = \{X_1\}$ and let $S_1$ be all the other nodes.
Then we have an example of case~(a) in Theorem~\ref{shares2-th}, and we can conclude that the optimum simultaneous evaluation of all the CQ's for the cycle of length $p$ gives equal shares to all the nodes except $X_1$, which gets a half share.

For a concrete example, assume $p=6$ and let the size of the data graph be $m=10^9$ edges. If we use $k$ = 500,000 reducers, then $X_1$ gets share 5 and the other five nodes get share 10 each.  The edges are replicated 10,000 times for each of the terms $E(X_2,X_3)$, $E(X_3,X_4)$, $E(X_4,X_5)$, and $E(X_5,X_6)$ (or the same terms with argument order reversed), while the edges are replicated 5000 times for the terms $E(X_1,X_2)$ and $E(X_1,X_6)$.  Thus, the total communication from mappers to reducers is $5\times10^{13}$.
Each reducer gets $10^8$ edges to deal with.
\end{example}

\subsubsection{More Details on Variable-Oriented Processing}
\label{examples-combine-subsec}

In this section, we first explain in detail how the techniques of \cite{AU10} are
applied in the general case to minimize the communication cost, and then
we focus again in regular graphs and give some more special cases with examples.

In the general case we use the techniques of \cite{AU10} as follows. We will build a number of Lagrangian equations to solve, where each equation equates any two
sums from a collection of sums. Each of the sums in the collection is constructed by
considering a node of sample graph $S$ and producing one term
for each edge adjacent to this node. The term for
the edge $(X,Y)$ is $1/(xy)$ or $2/(xy)$ depending
whether we will use both orientations of this edge or not
(where $x,y$ are the shares of the attributes $X,Y$ respectively).
Now we will write the equations for the general case of any sample
graph $S$ with $p$ nodes and then consider special cases  when the sample graph
$S$ is regular with degree $d$.

It is a convenience in the calculations if we view the edges
of $S$ that are to be used in both orientations to form a subgraph $H$ of $S$
and then write the equations separately for nodes that have all their adjacent
edges in $H$ or not.
Thus we define set $S_1$ whose nodes are adjacent only to edges in $H$
set $S_2$, whose nodes are
adjacent to both kinds of edges and set $S_3$ whose adjacent to only edges outside $H$. The sets of nodes  $S_1,S_2$ and $S_3$ is a partition of the nodes of $G$.

For node $i$  in $S_1$ we denote its share by $a_i$, for node $i$ in $S_3$ we denote its share by $b_i$ and for node $i$ in $S_2$ we denote its share by $z_i$. After writing the sums in the
collection we will investigate under what conditions all $a_i$'s  are equal to $a$ and all $b_i$'s are equal to $b$ and all $z_i$'s are equal to $z$. There are three kinds of sums:

\begin{itemize}
\item
The node $i$ is in $S_1$.  Then suppose $z_{i_j}, j=1,2,\ldots$ is the share of nodes adjacent to $i$ in $S_2$ and $a_{i_j}, j=1,2,\ldots$
is the share of nodes adjacent to $i$ in $S_1$ -- for simplicity we abuse notation using subscript $i_j$ for enumerating both
nodes in $S_2$ and $S_1$. Then the sum for this node is:

$$\frac{2}{a_ia_{i_1}}+ \frac{2}{a_ia_{i_2}} + \cdots + \frac{2}{a_iz_{i_1}}+ \frac{2}{a_iz_{i_2}} + \cdots $$
\item The node $i$ is in $S_3$. Then almost
 symmetrically with the first kind above, the sum for this node is:

$$\frac{1}{b_ib_{i_1}}+ \frac{1}{b_ib_{i_2}} + \cdots + \frac{1}{b_iz_{i_1}}+ \frac{1}{b_iz_{i_2}} + \cdots $$
\item The
 node $i$ is in $S_2$.  Then it has all three kinds of adjacent nodes. Thus we
denote by $z_{i_j}, j=1,2,\ldots$  the share of nodes adjacent to $i$ in $S_2$, by  $a_{i_j}, j=1,2,\ldots$ the share of nodes adjacent to $i$ in $S_1$ and by $b_{i_j}, j=1,2,\ldots$ the share of nodes adjacent to $i$ in $S_3$. The sum for this node is:

$$\frac{2}{z_ia_{i_1}}+ \frac{2}{z_ia_{i_2}} + \cdots + \frac{1}{z_ib_{i_1}}+ \frac{1}{z_ib_{i_2}} + \cdots	+\Sigma \frac{e}{z_iz_{i_1}}$$
where $e$ takes the value 2 or 1 depending whether it is an edge taken in both
orientations or only in one orientation.
\end{itemize}
Now suppose  that the sample graph is regular with degree equal to $d$ and that all $a_i$'s  are equal to $a$ and all $b_i$'s are equal to $b$ and all $z_i$'s are equal to $z$. Then the three kinds of the above sums turn into:

\begin{itemize}
\item The node $i$ is in $S_1$. Then let $d'$ be the number of nodes in $S_1$ that are adjacent to node $i$. Obviously $d-d'$ is the number of nodes in $S_2$ that are adjacent to node $i$ and there are no nodes in $S_3$ adjacent to $i$ because all edges adjacent to $i$ are taken in both directions. Then the sum is: $$\frac{2d'}{a^2}+\frac{2 (d-d')}{az}$$

\item The node $i$ is in $S_3$. Then let $d''$ be the number of nodes in $S_3$ that are adjacent to node $i$. Obviously $d-d''$ is the number of nodes in $S_2$ that are adjacent to node $i$ and there are no nodes in $S_1$ adjacent to $i$ because all edges adjacent to $i$ are taken in both directions. Then the sum is: $$\frac{d''}{b^2}+ \frac{d-d''}{bz} $$

\item The node $i$ is in $S_2$. Now node $i$ can be adjacent to any of the three kinds of nodes. Let  $d_{11}$ nodes in $S_1$ be adjacent to node $i$, $d_{12}$ nodes in $S_3$ be adjacent to node $i$. The sum is:
    $$\frac{2d_{11}}{za}+ \frac{d_{12}}{zb}+ \frac{e}{z^2}$$
    where $d-d_{11}-d_{12}\leq e \leq 2(d-d_{11}-d_{12})$.
\end{itemize}

Obviously the $d'$, $d''$, $d_{11}$, $d_{12}$ and $e$ mentioned above may be
different for different nodes $i$ (we should have used a subscript $i$ but
for simplicity and since we will mostly focus later in cases
where they are not be dependant on the node
$i$, we dropped the subscript). We first show below that either at least two of
the $a,b,z$ are integer
multiples of each other or $d'$ and $d''$ do not
depend on the particular node $i$.
We proceed under the assumption that  $d'$ and $d''$ do not
depend on the particular node $i$.
Then we argue how the $d_{11}$, $d_{12}$ and $e$
are related with each other and then we consider special cases where
the $d_{11}$, $d_{12}$ and $e$ do not depend on the node $i$ either.

\begin{itemize}
\item Let us equate the two sums for arbitrary pair of nodes $i$ and $j$ in $S_1$:
$$\frac{2d_i'}{a^2}+\frac{2 (d-d_i')}{az}=\frac{2d_j'}{a^2}+\frac{2 (d-d_j')}{az}$$
Then the above equation implies that:
$$(d_i'-d_j')z=-a(d_i'-d_j')$$
Hence either $z=a$ or $d_i'=d_j'$.
\item Symmetrically, we conclude that either $z=b$ or $d''$
is independent of the node $i$.
Before we proceed to the third bullet, let us talk about the options offered so far.
We have four choices: a) $a=b=z$; this is not possible because then one sum is
equal to $2d/a^2$ and the other equal to $d/a^2$. b) Both $d'$ and $d''$ are independent
of node $i$; we will discuss about it in the rest of this subsection. c) $a=z$ and $d''$ is
independent of $i$ or $b=z$ and $d'$ is independent of $i$. In this case we can use
the sum just above and to solve for $b/a$ when equating it to $2d/a^2$ which is
what the sum in the first bullet gives for $a=z$. This will give us $b/a$ as a
function of $d$ and $d''$. Then we can use the equation that says that the
product of all shares is equal to $k$ to solve for $a,b,z$. This solution however
will be valid only in the case it computes every sum in the third bullet
below to $2d/a^2$.

We continue assuming that both $d'$ and $d''$ are independent
of node $i$.
\item Let us equate the two sums for arbitrary pair of nodes $i$ and $j$  in $S_2$:
$$\frac{2d_{i11}}{za}+ \frac{d_{i12}}{zb}+ \frac{e_i}{z^2}=
\frac{2d_{j11}}{za}+ \frac{d_{j12}}{zb}+ \frac{e_j}{z^2}$$
The above implies:
$$\frac{2}{a}(d_{i11}-d_{j11}) + \frac{1}{b}(d_{i12}-d_{j12}) +\frac{1}{z}(e_i-e_j)=0$$
The above means that for every pair of nodes the difference $e_i-e_j$ is a linear
combination of the differences $d_{i11}-d_{j11}$ and $(d_{i12}-d_{j12}$
\end{itemize}

Special cases of the third bullet above can be considered:
\begin{enumerate}
\item
Suppose $e_i=e_j=e$ for all pairs. This means that $S_2$ is an independent set.
This means that for any pair of nodes $i,j$ we have:
$$\frac{2b}{a}=- \frac{(d_{i12}-d_{j12})}{(d_{i11}-d_{j11}) }$$
\item Suppose $e_i=e_j=e=0$. Then  $d_{i12}=d-d_{i11}$ for all $i$.
Then either $a=2b$ or $d_{11}$ is independent of the node $i$.
If we take $a=2b$ this leads to impossibility easily.
Thus we assume $d_{11}$ is independent of the node $i$ and the sum of the third kind
is:
$$\frac{2d_{11}}{za}+ \frac{d-d_{11}}{zb}$$
Now we have only three sums in the collection. Thus we solve for $a,b,z$
remembering also that $abz=k$, where $k$ is the number of reducers.

First we express $z$ as a function of $a$ and $b$
by considering the sum for nodes in $S_3$ and the sum for nodes in $S_2$.
$$\frac{1}{z}(\frac{2d_{11}}{a}+\frac{d-d_{11}}{b}-\frac{d-d''}{b})=\frac{d''}{b^2}$$
hence
$$z= (\frac{2d_{11}}{a}+\frac{d-d_{11}}{b}-\frac{d-d''}{b})\frac{b^2}{d''}$$
We can also express $z$ using the sum for $S_1$ and the sum of $S_3$:
$$\frac{1}{z}(\frac{2 (d-d')}{a}-\frac{d-d''}{b})=\frac{d''}{b^2}-\frac{2d'}{a^2}$$
Multiplying the two last equations we have a quadratic equation to express $1/a$ in terms of $b$:

$$(\frac{2 (d-d')}{a}-\frac{d-d''}{b})=(\frac{2d_{11}}{a}+\frac{d-d_{11}}{b}-\frac{d-d''}{b})\frac{b^2}{d''}(\frac{d''}{b^2}-\frac{2d'}{a^2})$$
\end{enumerate}
We give closed forms for two examples of the last case above.
%%Foto: the rest should be ok but I will re-check later

\begin{example}
For an example if $d'=d''=d_{11}=\frac{d}{2}$ then we have:

$$(\frac{2 }{a}-\frac{1}{b})=(\frac{2}{a}+\frac{1}{b}-\frac{1}{b})b^2(\frac{1}{b^2}-\frac{2}{a^2})$$
We solve the above and get $\frac{a}{b}=2^{\frac{1}{3}}$. And taking into account the equation for $z$ we get: $z=b2^{\frac{2}{3}}$.
Putting $a^{s_1}b^{s_3}z^{s_2}=k$ (where $k$ is the number of reducers, and $s_1$, $s_2$, and $s_3$ are the cardinalities of the
sets  $S_1,S_2$ and $S_3$ respectively), we get $$b=k^{\frac{1}{s_1+s_2+s_3}}2^{-\frac{s_1+2s_2}{3(s_1+s_2+s_3)}}$$
Thus the replication per input tuple is (where $p$ is the number of nodes in sample graph $S$):

$$k\frac{p(d/2)(1+\frac{1}{2^{2/3}})2^{\frac{2(s_1+2s_2)}{3(s_1+s_2+s_3)}}}{k^{\frac{2}{s_1+s_2+s_3}}}$$

or, since $s_1+s_2+s_3=p$

$$k\frac{p(d/2)(1+\frac{1}{2^{2/3}})2^{\frac{2(p-s_3)}{3p}}}{k^{\frac{2}{p}}}$$

or by simplifying

$$kpd\frac{1+2^{2/3}}{2^{1+\frac{2s_3}{3p}}k^{\frac{2}{p}}}~~~~~~~~~~~~~~~~~~~Eq. (2)$$
\end{example}

\begin{example}
For another example let us assume that  the set $S_2$ is an independent set and also
is such that each edge
of the graph $S$ contains a node in $S_2$. Then the solution
that optimizes the communication cost is the following.
The nodes
in $S_1$ take share equal to $a$ and the nodes in
$S_3$ take share equal to $a/2$. The nodes in $S_2$ take share equal to $z=a$.
Then $a=k^{1/p}2^{s_3/p}$.
The replication per input tuple in this case is:
%$$kpd(k^{1/(s_1+s_3)}\sqrt{2}^{s_1/(s_1+s_3)})^{-2}$$
%or, since $s_1+s_2+s_3=p$
%and $ds_2=pd/2$ (since all edges have an endpoint in $S_2$)
$$kpd\frac{2}{2^{\frac{2s_3}{p}}k^{\frac{2}{p}}} ~~~~~~~~~~~~~~~~~~~~~~~~~~~~~~~~~~~~Eq. (3)$$
\end{example}
Equations (2) and (3) show how the communication cost may vary in different situations.

\subsection{Advantage of Variable-Oriented Processing}
\label{combine-app}

In this section, we prove that it is always more efficient to combine all CQ's for a sample graph than it is to evaluate them separately or in groups.
%Finally, we discuss a hybrid approach (see Section~\ref{hybrid-subsect}) where we use uniform shares for the variables and order the nodes of the data graph by the bucket into which they hash, as we did in Section~\ref{improved-subsect}.

In Example~\ref{sq-cc-ex} we explained how to compute the communication cost
when each reducer produces the portion of the result for each of the CQ's~-- that
portion is determined by the list of buckets associated with the reducer.
The following theorem says that, in order to minimize the communication cost
we should use one hash function for the entire group of CQ's. Remember that
each CQ in the group has the same relational subgoals, with arguments in
different order.  Thus when we combine many CQ's to be executed in the same
reducer, each term in the communication cost is a sum of the same terms only
with different coefficients (which are either 1 or 2).

\begin{theorem}
\label{group-th}
Let $S$ be a sample graph and ${\mathcal{Q}}$ be a group of queries that produce the instances of $S$ in a data graph $G$; i.e.,  each query in ${\mathcal{Q}}$ has a number of relational subgoals, each subgoal corresponding to an edge of $S$ (hence the relational subgoals in queries in ${\mathcal{Q}}$ differ only in the order of the variables in their arguments).
Then the communication cost of computing all queries in the group ${\mathcal{Q}}$ by breaking down the group in several subgroups is greater than or equal to the communication cost of computing the group ${\mathcal{Q}}$ by combining all CQ's in one whole.
\end{theorem}

\begin{proof}
Because each query in the group has the same relational subgoals up to argument reordering,  the expression that gives the communication cost (hereafter in this proof referred to as the ``cost'') for each subgroup has the same terms but with different coefficients.
Moreover each coefficient is either 1 or 2. Finally each term is a product of variables, each variable being constrained to be greater than or equal to 1. Let $E$ be the expression that corresponds to the cost of ${\mathcal{Q}}$ and $E_i$ be the expression
that corresponds to the cost of a subgroup ${{\mathcal{Q}}_{i}} \subseteq {\mathcal{Q}}$.

First we prove the following. Suppose cost expression $E_1$ differs from cost expression $E_2$ in that the terms in $E_2$ with coefficient equal to 2 is a superset of those terms in $E_1$.  Then the following is true:

Claim (*) : $\min{E_1} \leq \min{E_2}$.

In proof of the claim let ${\mathcal{A}}$ be an assignment  of values to the variables in expression $E_2$   that minimizes $E_2$.
If we use the same assignment ${\mathcal{A}}$ in $E_1$, we get something smaller than
what we get by using this assignment ${\mathcal{A}}$ in $E_2$, hence smaller than
$\min{E_2}$. Since the assignment ${\mathcal{A}}$ does not necessarily minimize $E_1$, $\min{E_1}$ could be even smaller, hence $\min{E_1} \leq \min{E_2}$.

Now let:

\begin{itemize}
\item $OPT_{{{\mathcal{Q}}_{all}}}$ be the minimum cost for computing all the queries in ${\mathcal{Q}}$ as a whole.

\item $OPT_{{{\mathcal{Q}}_{single}}}$ be the minimum cost for computing a single query, i.e,., all coefficients are equal  to 1.
\item $OPT_{{{\mathcal{Q}}_{i}}}$ be the minimum cost for computing all the queries in ${{\mathcal{Q}}_{i}}$ as a whole, where ${{\mathcal{Q}}_{i}} \subseteq {\mathcal{Q}}$
\end{itemize}

Then the following holds:

$$OPT_{{{\mathcal{Q}}_{all}}} \leq 2OPT_{{{\mathcal{Q}}_{single}}} \leq OPT_{{{\mathcal{Q}}_{1}}} + OPT_{{{\mathcal{Q}}_{2}}}$$

The first inequality above is proven by observing that the cost expression for twice the cost of a single query compared to the
cost expression for ${\mathcal{Q}}$, has the property of Claim~(*), i.e., the terms with coefficient equal to 2 in the former is a superset
of those in the latter. Hence the first inequality is a consequence of Claim~(*). The second inequality is also a consequence of Claim~(*),
since it comes from two inequalities $OPT_{{{\mathcal{Q}}_{single}}} \leq OPT_{{{\mathcal{Q}}_{1}}}$ and $OPT_{{{\mathcal{Q}}_{single}}} \leq OPT_{{{\mathcal{Q}}_{2}}}$.
\end{proof}

\subsection{Bucket-Oriented Processing}
\label{bucket-oriented-subsect}

While the method of Section~\ref{combine-cq-subsect} determines the optimal number of buckets for each variable, this approach uses the same number of buckets, $b$, for each of the variables in each CQ.  We thus lose the opportunity to optimize this number of buckets, but the compensating advantage is that each edge is distributed among the reducers in only one orientation.  Which method is better depends on how far from optimal the choice of equal numbers of buckets is, and on how many subgoals among the set of CQ's have their arguments in both directions.

The bucket-oriented approach does the following:

\begin{enumerate}
\item Create a reducer for each nondecreasing sequence of $p$ bucket numbers in the range 1 to $b$.

\item For each edge $(u,v)$, hash $u$ and $v$ to buckets, using the hash function $h$.  To determine which reducers get this edge, form a multiset of integers, starting with $h(u)$ and $h(v)$.  Then, add $p-2$ integers in the range 1 to $b$.   These integers may duplicate $h(u)$ and $h(v)$ as well as each other.  Sort the multiset to get a nondecreasing list, which corresponds to exactly one reducer.  These are the reducers that receive a copy of the edge $(u,v)$.

\item Let each reducer evaluate each of the CQ's for the given sample graph, using the edges it is given.  Since every CQ has a total order of the variables, each solution for the CQ will be discovered by exactly one of the reducers.
\end{enumerate}

While we cannot directly compare the bucket-oriented and vari\-able-oriented methods, we can at least claim that the bucket-ori\-ented method beats the generalization of Partition for arbitrary sample graphs. The advantage, however, decreases as $p$ increases.

To this end, we first work as in Theorem~\ref{order-buckets-th}, and count the number of useful reducers by comparing them to certain binary strings. Thus, we can show that the number of reducers used by an application of the bucket-oriented method, using $b$ buckets for a sample graph of $p$ nodes is $\binom{b+p-1}{p}$.
Similarly, we can count the number of reducers that receive each edge, which is equal to $\binom{b+p-3}{p-2}$. For large $b$, this count is approximately $b^{p-2}/(p-2)!$.

Let us compare this number with what we get by generalizing the Partition Algorithm. If instead we partition the nodes into $b$ buckets and use reducers that correspond to sets of $p$ groups (the obvious generalization of the Partition Algorithm), then edges going between nodes in two different groups are sent to $\binom{b-2}{p-2}$  reducers.  However, an edge going between nodes of the same group are sent to $\binom{b-1}{p-1}$ reducers.  Since $1/b$th of the edges are of the latter kind, the average number of reducers receiving an edge is, for large $b$, approximately
$$b^{p-2}\Bigl(\frac{1}{(p-2)!} + \frac{1}{(p-1)!}\Bigr)$$
Thus, the ratio of the communication cost per edge for generalized Partition, compared with that of the bucket-oriented algorithm is, for large $b$, equal to $1+\frac{1}{p-1}$.
This ratio is always greater than 1, although it approaches 1 as $p$, the number of nodes in the sample graph, gets large.

\section{Conjunctive Queries for Cycles}
\label{cyc-sect}

In this section we consider an algorithm for finding all occurrences of the cycle $C_p$ of length $p$.  The strategy is based on the orientation of the edges.  Intuitively, when we start with all orders of the nodes, use the automorphisms to reduce the number, and then further reduce the number of CQ's by clustering according to edge orientations, we do not avoid many of the possible orientations.  If we start with the orientations only, then we can use the automorphisms effectively to cut down the number of orientations that actually need CQ's.  However, the effect of automorphisms on edge orientations is more complex than their effect on node orders, so it is only in special cases such as cycles that we can get general rules for selecting CQ's by starting with orientations.

\begin{figure}
\centerline{\includegraphics[width=0.37\textwidth]{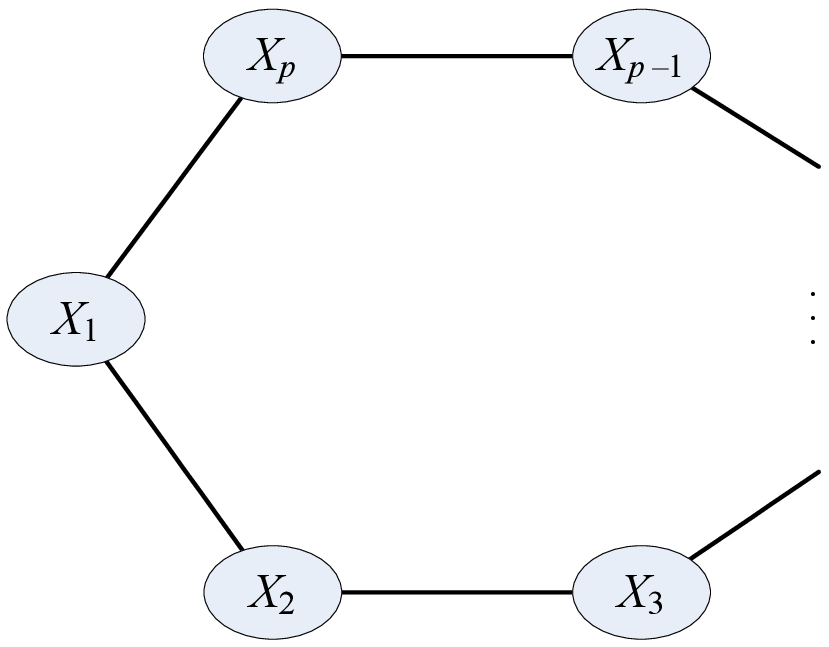}}
\caption{A cycle of $p$ nodes}
\label{cycle-fig}
\end{figure}

Imagine cycle $C_p$ with nodes $X_1,X_2,\ldots,X_p$ arranged in a circle, with $X_i$ counterclockwise of $X_{i-1}$, as suggested in Fig.~\ref{cycle-fig}.
If $X_{i-1}<X_i$ we shall say the edge $(X_{i-1},X_i)$ is an {\em up edge} (designated by $u$) and otherwise it is a {\em down edge} (designated $d$).  Likewise, if $X_p>X_1$ then the edge $(X_1,X_p)$ is an up edge, and otherwise it is a down edge.

Every cycle can be oriented so that $X_1$ is lower than its neighbors.  For example, we could pick $X_1$ to be the lowest node on the cycle, but often there are other choices for $X_1$ as well.  In general, an orientation of the edges counterclockwise around the cycle can be described by the runs of up and down edges.  This sequence must begin with a run of up edges and end with a run of down edges, because of our assumption about $X_1$.  The sum of the run lengths is $n$, and there must be an even number of runs, because they begin with $u$ and end with $d$.

\begin{example}
\label{penta-ex}
Consider the pentagon $C_5$.  The possible sequences of run lengths are 14, 23, 32, 14, 1112, 1121, 1211, and 2111.  These are all the sequences of positive integers that sum to five and have even length.  The sequence 14 corresponds to the orientation of edges where, starting with $X_1$ and proceeding counterclockwise, we have orientations $udddd$.  That is, $X_1<X_2$, but $X_2>X_3>X_4>X_5>X_1$.  Similarly, the other seven sequences of runs correspond to orientations $uuddd$, $uuudd$, $uuuud$, $ududd$, $uduud$, $uddud$, and $uudud$.
\end{example}

\subsection{Automorphisms and Run Sequences}
\label{auto-run-subsect}

The cycle $C_p$ has an automorphism group of size $2p$.  This group is the product of the group of cyclic shifts ($p$ elements) and the group of two elements ``flip'' and ``don't flip'' (the identity).  Some run sequences are transformed into other sequences by these automorphisms.
Because we insist that $X_1$ be lower than its neighbors, not every cyclic shift corresponds to another run sequence.  However, if we rotate the run sequence by two (which may correspond to rotating the cycle by more than two positions), we get another node as $X_1$, and that node will also be less than its neighbors.  To avoid double-counting of cycles, we want to eliminate a run sequence if it is a cyclic shift of another run sequence by an even number of positions.

\begin{example}
\label{penta2-ex}
Thus, in Example~\ref{penta-ex}, the orientations $ududd$ and $uddud$ are equivalent; each is a cyclic shift by two runs of the other.  Thus, the CQ for either produces exactly the same instances of $C_5$ that the other does.  Likewise, $uduud$ and $uudud$ are equivalent, and we can use either one.  Let us therefore eliminate $uddud$ and $uudud$.
\end{example}

The automorphism in which we flip the cycle also allows us to eliminate some of the run sequences.  Flipping reverses the sequence of run lengths.  Its effect on a sequence of $u$'s and $d$'s is twofold:

\begin{enumerate}

\item
The sequence of $u$'s and $d$'s is reversed.

\item
Then each $u$ is replaced by $d$ and vice-versa.

\end{enumerate}

\begin{example}
\label{penta3-ex}
The six orientations that remain after using the cyclic shifts of Example~\ref{penta2-ex} are $udddd$, $uuddd$, $uuudd$, $uuuud$, $ududd$, and $uduud$.
If we flip $udddd$, we reverse it to get $ddddu$ and then exchange $u$'s and $d$'s to get $uuuud$.  That is, we can eliminate $uuuud$ in favor of $udddd$.  Simularly, the flip of $uuddd$ is $uuudd$, so we can eliminate the latter.  Finally, the flip of $ududd$ is $uudud$.  The latter was already found to produce the same cycles as $uduud$, so we know that $uduud$ provides no cycles that $ududd$ does not provide.  There are thus only three CQ's needed to find all pentagons, those corresponding to orientations $udddd$, $uuddd$, and $uduud$.  These CQ's are, respectively:

\begin{center}
\begin{tabular}{l}
$E(X_1,X_2)$ \& $E(X_3,X_2)$ \& $E(X_4,X_3)$ \& $E(X_5,X_4)$ \& $E(X_1,X_5)$\\ \& $X_1<X_2$ \& $X_3<X_2$ \& $X_4<X_3$ \& $X_5<X_4$ \& $X_1<X_5$\\
$E(X_1,X_2)$ \& $E(X_2,X_3)$ \& $E(X_4,X_3)$ \& $E(X_5,X_4)$ \& $E(X_1,X_5)$\\ \& $X_1<X_2$ \& $X_2<X_3$ \& $X_4<X_3$ \& $X_5<X_4$ \& $X_1<X_5$\\
$E(X_1,X_2)$ \& $E(X_3,X_2)$ \& $E(X_3,X_4)$ \& $E(X_4,X_5)$ \& $E(X_1,X_5)$\\ \& $X_1<X_2$ \& $X_3<X_2$ \& $X_3<X_4$ \& $X_4<X_5$ \& $X_1<X_5$\\
\end{tabular}
\end{center}

Notice that if we use the methods of Section~\ref{sg-cq-sect} we wind up with seven CQ's rather than the three above.  That is, there are 120 orders of five nodes.  The automorphism group of the pentagon has size 10, so we start with 12 CQ's.  However, if we choose these CQ's to all satisfy the constraints that $X_1$ is smallest and $X_2<X_5$, then the CQ's group into seven orientations of the remaining edges $(X_2,X_3)$, $(X_3,X_4)$, and $(X_4,X_5)$.  Note that one of the eight orientations is impossible, because we cannot have $X_2<X_5<X_4<X_3<X_2$.
\end{example}

While we always can eliminate a run sequence that produces the same instances as some other run sequence, there are some run sequences that are automorphic to themselves.  We cannot eliminate a sequence in favor of itself, so we are forced to find some other way to eliminate duplication.   The number of times each cycle will be discovered by the CQ for a sequence is the number of flips and cyclic shifts (including the identity) that leave the sequence unchanged.  We can avoid discovering a cycle more than once by adding the inequalities that make $X_1$ the smallest of all nodes, and also the inequality $X_2<X_p$ to prevent a cycle and its flip from both being recognized.  The problem and its solution can be seen by examining the hexagon, $C_6$.

\begin{example}
\label{hex-ex}
For the hexagon, there are five run sequences of length 2; these are 15, 24, 33, 42, and 51.  The last two are obviously reversals of the first two, so we can eliminate them.  But 33, which corresponds to the orientation $uuuddd$, will produce each hexagon that it produces twice, as any matching hexagon can be flipped.

There is only one run sequence of length 6: 111111, or $ududud$.  This sequence matches each hexagon that it matches at all six times, corresponding to zero, one, or two rotations of 120 degrees and/or flipping.

There are ten sequences of four runs.  We can have one run of 3 and three runs of 1.  But of these, 1113 is the reverse of 3111 and 1311 is the reverse of 1131, so only 1113 and 1131 need be considered.   Other sequences of four runs have two 1's and two 2's.  There are six such sequences, but when we eliminate reversals and rotations by two positions, we are left with only 1122, 1212, and 1221.  There are thus seven sequences for which we must write CQ's: the three just mentioned plus 15, 24, 33, and 111111.  All but the last two are straightforward.  For 33, we need to force $X_2<X_6$ to prevent flipping.   For 111111, we need to force $X_1<X_3$, and $X_1<X_5$ to prevent rotation by 120 or 240 degrees, and we need to force $X_2<X_6$ to prevent flipping.  The resulting CQ is

\begin{center}
\begin{tabular}{l}
$E(X_1,X_2)$ \& $E(X_3,X_2)$ \& $E(X_3,X_4)$ \& \\
$E(X_5,X_4)$ \& $E(X_5X_6)$ \& $E(X_1,X_6)$ \& \\
$X_1<X_2$ \& $X_3<X_2$ \& $X_3<X_4$ \& $X_5<X_4$ \& $X_5<X_6$ \& \\
$X_1<X_6$ \& $X_1<X_3$ \& $X_1<X_5$ \& $X_2<X_6$\\
\end{tabular}
\end{center}
\end{example}

\subsection{Algorithm for Finding Cycles Using Runs and Orientation}

The algorithm for finding uniquely cycles of length $p$ that formalizes the above examples is the following:

\begin{enumerate}

\item Find all bags (i.e., multisets) of an even number of positive integers that sum to $p$.

\item
For each bag find all permutations of its elements
and form set $S_1$ of permutations. From $S_1$ delete permutations to form its subset $S$, so that in $S$ no permutation is a nontrivial cyclic shift, with optional flip, of another.

\item
For each permutation in $S$  (generated in the previous
step) create the corresponding pattern of $u$'s and $d$'s; i.e., as you read the permutation, replace every integer by that number of $u$'s or $d$'s,, starting with $u$'s and alternating $u$'s and $d$'s.
This sequence of $u$'s and $d$'s tells us the order of arguments for the relational subgoals of the CQ for this permutation.  The arithmetic subgoals enforce the relationship between adjacent nodes of the cycle, as are implied by the $u$'s and $d$'s.  We then modify these CQ's as follows.

\item
For each CQ created in the previous step do:

\begin{enumerate}

\item If this CQ is not a palindrome and has no nontrivial periodicity, do nothing.

\item If this CQ is a palindrome then add the inequality $X_2<X_p$.

\item If this CQ has nontrivial periodicity, then add a number of inequalities
(equal to $p$ divided by the period, i.e., the length of the
smallest repeated string) that say that
$x_1$ is less than any of the other positions that are also less than
both neighbors.

\end{enumerate}
\end{enumerate}

%The following theorem establishes that when the CQ's created by the algorithm are applied to a data graph, each cycle is discovered once.

\begin{theorem}
\label{auto-once-thm}
The group of CQ's created by the algorithm is such that, when it is applied to a data graph, each cycle is discovered once.
\end{theorem}

\begin{proof}
Suppose two different CQ's compute the same cycle. This cannot
happen because then they have to be a cyclic shift or a flip of each other,
and the algorithm never includes two such CQ's.
So, if a cycle is discovered twice, it must be by the same CQ.
In this case one assignment of nodes to variables of the CQ follows around the cycle a string  $AB$ (where $A,B$ are strings of $u$'s and $d$'s)
and the other follows around the cycle a string $BA$. Since it is the same CQ, we
know $AB=BA$. Thus there is a string $W$ and integers $i,j$
such that $A=W^i$ and $B=W^j$. Hence there is periodicity in the query, and we show below that the extra in\-equal\-i\-ties will take care of the uniquness.

There are two cases:

\begin{enumerate}

\item
The CQ is neither a palindrome nor has periodicity.
Since it is not a palindrome we cannot start computing the subgoals
in the opposite direction. Since it has no periodicity, we cannot
start computing the subgoals from any other node of the cycle either.

\item
It is either a palindrome or has periodicity, or both.
If it is a palidrome with no periodicity, the CQ can only compute the cycle starting from the same node but going the opposite direction. This cannot happen because
we put the extra inequality to force the assignment to choose the smaller
of the two neighbors of the starting node.
In the second case, the CQ can start also from several other nodes of
the cycle. The extra inequalities involving $X_1$ force $X_1$ to be the globally lowest node in the cycle instance.

\end{enumerate}
\end{proof}

Now we claim that the group of CQ's found by our algorithm is minimum.
Below, first we prove that it is minimal and then we prove that it is minimum (i.e., unique up to renaming of variables and order of subgoals) too.
Suppose $C$ is a cycle in the data graph.  Define the {\em characteristic sequence} of edge orientations for $C$ by starting at the node with the minimum
value and proceeding around the cycle by visiting next the neighbor of the starting node
with the smaller value. We produce the sequence by reporting
on the orientation ($u$ or $d$) of each edge we are visiting. We claim that each cycle
can be discovered by a CQ that either corresponds to its characteristic sequence
or to a cyclic shift or to a flip of a cyclic shift of the characteristic sequence.
This is so because we can start computing the query by mapping its first variable
to the lowest node of the cycle and either go clockwise or counter-clockwise.
Now, each group of CQ's contains only one CQ that is either the characteristic sequence
of $C$ or a flip or a flip of a cyclic shift or a cyclic shift of the characteristic sequence. Hence if we remove this CQ from the group, then $C$ will not be
discovered. Thus the group of CQ's we construct using our algorithm is minimal.

This group is also minimum for the following reason.
For ease of reference let us call each CQ that comes from a certain orientation and includes only inequalities
that are implied by those orientations a {\em basic CQ for $C$.}
If we take all groups of CQ's such that each group is closed
under flipping and cyclic shift, then we get a partition of all possible CQ's (i.e., all CQ's that can be formed
by taking a given orientation of the edges). Any group of CQ's that discovers all cycles
contains at least one from each subset of this partition, otherwise a cycle may be missed.
Thus the algorithm constructs a group that is minimum, since it only includes exactly one basic CQ from each subset of the partition.

\subsection{An Upper Bound on the Number of CQ's and Some Examples}

When the sample graph is a cycle of length
$p$ we can assume that the $2^p-2$ sequences of
$u$'s and $d$'s (the minus 2 is due to the fact
that all-$u$ and all-$d$ are not included) is divided by $p$
because only one representative of the $p$ cyclic shifts is included. It can be further divided by 2 because flips are
not included. Thus we get an upper bound of $(2^p-2)/(2p)$.
However this upper bound holds only in the cases
where there are no
sequences with nontrivial periodicity in the group.
If a sequence
has periodicity greater than one, then one sequence does not exclude as many as $p-1$ other sequences,
because  two distinct cyclic shifts of the sequence may result in the same sequence.
E.g., for the sequence $uuud$, all three other sequences that are created by a cyclic
shift are pairwise distinct, but for the sequence $udud$, we get only one other distinct
sequence that comes from a cyclic shift of one position, and this is the $dudu$.
Thus we call the upper bound  $(2^p-2)/(2p)$ a {\em conditional} upper
bound, since it holds only if all sequences have periodicity equal to 1.
Thus   $(2^p-2)/(2p)$ is an unconditional
upper bound (which is tight as we show below)
only when $p$ is a prime.

\begin{example}
Here are some examples for various $C_p$.

\begin{itemize}
\item $p=6$.  The conditional upper bound is (64-2)/12=5.17, but actually there are the following 7 CQ's in the group, from the following run sequences:
 111111,
 1122,
 1212,
 1113,
 15,
 24,
 33.

We observe that 111111 has periodicity 6 (3 multiplied by 2 to account for the flips too).
1212 has periodicity 2, and 33 has periodicity 2.  Summing up the
periodicities we get $6+2+2=10$.
We also observe that there are two more palindromes: 2112 and
1221 (both from cyclic shift of 1122); each accounts for 6.
Summing up periodicities and palindromes: $10+6+6=22$.
Thus the upper bound
in this case is $(64+22-2)/12=7$, which is actually equal to the
minimum number of CQ's in a group.

\item $p=7$.  The conditional upper bound is $126/14=9$, and since 7 is a prime,
the minimum number of CQ's in a group is indeed equal to 9.
The underlying run sequences are:
 111112,
 1123,
 1132,
 1222,
 1213,
 1114,
 16,
 25,
 34.

%\item $n=8$.  The conditional upper bound is $(2^n-2)/(2n)=15.89$ but we have more than that, we have 17 CQ's in the group, with the following run sequences:
%11111111,
%111113,
%111122,
%111212,
%112112,
%1124,
%1214,
%1133,
%1313,
%1232,
%2132,
%2222,
%1115,
%17,
%26,
%35,
%44.
%
%However observe that 272/16=17 and 272 comes from
%$18+2^8-2$ where 18 comes from summing the periodicity (including the flip in the first and last of the sequences below) of the following periodic run sequences:
%11111111 (contributes 8 to the 18),
%112112 (contributes 2),
%1313 (contributes 2),
%2222 (contributes 4),
%44 (contributes 2).
\end{itemize}
\end{example}

\section{Map-Reduce Computation Cost}
\label{comp-cost-sect}

We now turn to the second important measure of the quality of a map-reduce algorithm~-- the total computation cost at the mappers and reducers.  In each of the algorithms discussed, the computation at the mappers is proportional to the communication cost, so we shall ignore the mappers and focus on the computation at the reducers.  This cost is polynomial in the size of the data graph, but the degree of the polynomial can be large, and the critical issue, to be addressed in the balance of this paper, is how low can we make the degree of the polynomial.

\subsection{Convertible Algorithms}
\label{conv-alg-subsect}

Each of the methods we have described depends on a serial algorithm for finding instances of the same sample graph.  This algorithm is used at each reducer and is applied to a smaller graph.
However, the relationship between the number of edges and the number of nodes in the graphs at each reducer generally differs from the node/edge relationship for the entire graph.

\begin{definition}
A {\em mapping scheme} is a function from input elements to sets of key-value pairs.
\end{definition}

\begin{definition}
Let us call a serial algorithm $\cal A$ {\em convertible} with respect to a given mapping scheme if, given random input, when the algorithm is run at each reducer the total computation cost at the reducers is, with high probability, proportional to the running time of $\cal A$ run on a single machine.  The constant of proportionality may depend upon characteristics of the algorithm but not on the number of reducers.
%\footnote{We do not assert that every mapping scheme can use a candidate for convertibility; it may be that what is done by the reducers bears no relationship to a serial algorithm for the problem being solved. However, for a simple example where the idea makes sense but the algorithm is not convertible, consider binary search of a linear list of length $n$. If the mapping scheme partitions the sorted list into $k$ consecutive blocks and passes each to a reducer, then the serial algorithm takes time $O(\log n)$.  But the reducers, able to do nothing but binary search on a sorted list of length $n/k$, take total time $k\log(n/k)$.  The ratio $\bigl(k\log(n/k)\bigr)/\log n$ grows with $k$, so binary search is not convertible with respect to this mapping scheme.}
\end{definition}

For algorithms that enumerate instances of a sample graph, there is only one mapping scheme that we have considered or will consider.
Assume that we hash nodes of the sample graph into $b$ buckets, and the reducers correspond to lists of bucket numbers, one for each node of the sample graph $S$.
Edges are sent from the mappers to those reducers whose lists include both nodes of the edge.  Then the probability that a node appears in the graph of a given reducer is $O(1/b)$; the constant of proportionality is approximately the number of nodes in the sample graph.  The probability that an edge appears in the sample graph at a given reducer is $O(1/b^2)$, since both its nodes must be hashed to buckets in the list for the reducer.  Since data is random, skew is limited, and with high probability the reducers all get within a constant factor of the average numbers of nodes and edges.  Finally, assume that the best serial algorithm for finding all instances of $S$ on a graph of $n$ nodes and $m$ edges has running time $O(n^{\alpha}m^{\beta})$ for some constants $\alpha$ and $\beta$.

The number of reducers $k$ is $O(b^p)$, where $p$ is the number of nodes of $S$.
The computation performed by any reducer on a graph of $O(n/b)$ nodes and $O(m/b^2)$ edges requires time
$$O\bigl((n/b)^{\alpha}(m/b^2)^{\beta}\bigr)$$
Thus, the total work at all the reducers is $O\bigl(b^p (n/b)^{\alpha}(m/b^2)^{\beta}\bigr)$, or simplifying $O(b^{p-\alpha-2\beta}n^{\alpha}m^{\beta})$.  Put another way, the total work at the reducers is on the order of the work of the serial algorithm on the original graph times $b^{p-\alpha-2\beta}$.  If this exponent is positive, then the work at the reducers exceeds the work of the serial algorithm.  However, when the exponent is nonpositive, we can conclude:

\begin{theorem}
\label{comp-cost-th}
If the best serial algorithm for finding all instances of a sample graph $S$ runs in time $O(n^{\alpha}m^{\beta})$ on a graph of $n$ nodes and $m$ edges, and $\alpha+2\beta$ is no less than the number of nodes of $S$, then there is a convertible algorithm for finding all instances of $S$.
\end{theorem}

\begin{example}
\label{comp-tri-ex}
Observe that for triangles, $p=3$, $\alpha=0$, and $\beta=3/2$, so the condition Theorem~\ref{comp-cost-th} holds.  We pointed out this observation of \cite{SV11} in Section~\ref{partition-subsect}.
\end{example}

\subsection{Decomposition of Sample Graphs}
\label{sg-decomp-subsect}

Let us call a serial algorithm with running time $O(n^{\alpha}m^{\beta})$, where $\alpha\ge 0$ and $\beta\ge 0$, an $(\alpha,\beta)$-algorithm.
Given that $p$, the number of nodes of the sample graph is a constant, any computation that depends only on the size of the sample graph, and not on the data graph, can be ignored when discussing $(\alpha,\beta)$-algorithms.
An important consequence of this observation is that it is possible to decompose sample graphs, and the algorithms for discovering instances of the subgraphs can then be combined in a way that preserves convertibility.

In what follows, we shall assume that the data graph is preprocessed so that there is an index on pairs of nodes that lets us determine in $O(1)$ time whether there is an edge between any two given nodes.  This index can be constructed in time $O(m)$, where $m$ is the number of edges of the data graph, and surely any algorithm for finding instances of a sample graph will at least look at each edge of the data graph.  Thus, the existence of this index will be assumed and the time to construct it can be ignored.

\begin{lemma}
\label{a-b-lemma}
Let $S$ be a $p$-node sample graph.  Partition the nodes of $S$ into two sets of $p_1$ and $p_2$ nodes, and let $S_1$ and $S_2$ be the subgraphs of $S$ generated by these two sets of nodes.  If $S_i$ has an $(\alpha_i,\beta_i)$-algorithm for $i=1,2$, then $S$ has an $(\alpha_1+\alpha_2,\beta_1+\beta_2)$-algorithm.
\end{lemma}

\begin{proof}
Use the two given algorithms to enumerate all instances of $S_1$ and $S_2$ in the data graph.
As there can be no more in\-stan\-ces of a subgraph than the running time of the algorithm that enum\-er\-ates them, there are $O(n^{\alpha_1}m^{\beta_1})$ instances of $S_1$ and  $O(n^{\alpha_2}m^{\beta_2})$ instances of $S_2$. Therefore, the number of pairs of these instances is  $O(n^{\alpha_1+\alpha_2}m^{\beta_1+\beta_2})$.  For each pair of instances:

\begin{enumerate}

\item
Check that the nodes of the two instances are disjoint.

\item
Check that for each edge in $S$ that connects a node of $S_1$ with a node of $S_2$, the edge between the corresponding nodes of the two instances exists in the data graph.

\item
Check that each instance of $S$ is generated only once.

\end{enumerate}

The work of the above steps for any one pair of instances depends only on $p$, the number of nodes of $S$.  Step~(1) clearly takes time $O(p)$.  The index allows Step~(2) to be carried out in $O(1)$ time per edge of $S$.

Step~(3) is a little trickier.  Whenever we generate an instance, we need to check that it is lexicographically first among all the ways that this instance can be generated from instances of $S_1$ and $S_2$.  To see how this step can be carried out, assign an order to the nodes of the data graph $G$.  Once we have identified an instance $H$ of $S$ in $G$, order the nodes of that instance according to the order for $G$.  Form a string of 1's and 2's, where the $i$th position of the string is 1 if the $i$th node of $H$ came from the instance of $S_1$ and 2 otherwise.  Now, consider all other possible assignments of the nodes of $H$ to the nodes of $S$, and construct their strings in the same way.  Only if this construction of $H$ is the lexicographically first among all these strings do we now generate instance $H$.  Otherwise, $H$ will be generated when we consider some other pair of instances of $S_1$ and $S_2$.

Thus, the total work is proportional to the number of pairs of instances plus the time taken by the algorithms that enumerate the instances of $S_1$ and $S_2$.  Since we assume the $\alpha$'s and $\beta$'s are nonnegative, the latter algorithms each take no more time than the upper bound on the number of pairs, which is $O(n^{\alpha_1}m^{\beta_1})$ times $O(n^{\alpha_2}m^{\beta_2})$.  Thus the entire algorithm for finding instances of $S$ takes time $O(n^{\alpha_1+\alpha_2}m^{\beta_1+\beta_2})$.
\end{proof}

\begin{theorem}
\label{conv-th}
Let sample graph $S$ be partitioned into $S_1$ and $S_2$ as in Lemma~\ref{a-b-lemma}.
Then if $S_1$ and $S_2$ have convertible algorithms with running times $O(n^{\alpha_i}m^{\beta_i})$ for $i=1,2$, then $S$ has a convertible algorithm with running time of the form $O(n^{\alpha}m^{\beta})$.
\end{theorem}

\begin{proof}
By Lemma~\ref{a-b-lemma}, $S$ has an algorithm with running time $O(n^{\alpha_1+\alpha_2}m^{\beta_1+\beta_2})$.
Let $S$ have $p$ nodes, and let $S_i$ have $p_i$ nodes for $i=1,2$.
Then $p_i\le\alpha_i+2\beta_i$, for $i=1,2$.
Thus $p=p_1+p_2\le (\alpha_1+\alpha_2) + 2(\beta_1+\beta_2)$.
Hence, the algorithm for $S$ constructed by Lemma~\ref{a-b-lemma} is convertible.
\end{proof}

\begin{example}
\label{decomp-ex}
\cite{Alon81} shows that if a sample graph $S$ can be decomposed into

\begin{enumerate}

\item
Pairs of nodes connected by an edge, and

\item
Odd-length cycles (possibly with additional edges between nodes of the cycle),

\end{enumerate}
then in the worst case the data graph has $\theta(m^{p/2})$ instances of $S$.
Recursive applications of Lemma~\ref{a-b-lemma} and Theorem~\ref{conv-th} let us show that there is a matching serial algorithm for every such sample graph, and moreover, the serial algorithm can be converted to a map-reduce algorithm with the same computation time.
We have only to exhibit serial algorithms for the basis cases (edges and odd-length cycles).
A pair of nodes with an edge is very easy; just enumerate the edges in the data graph.
That enumeration takes time $O(m)$, which is $O(m^{p/2})$ for $p=2$.
Thus, there is a $(0,1)$-algorithm, and by Theorem~\ref{comp-cost-th} this algorithm is convertible.  The case of odd cycles is much trickier, and we give the proof of the existence of an $O(m^{p/2})$ algorithm when the sample graph is a cycle of odd length $p$ in Corollary~\ref{cor-oddenum}.
\end{example}

\section{Optimal Serial Algorithms for General Sample Graphs}
\label{serial-alg-sect}

In this section, we show that the bounds of \cite{Alon81} can be met with concrete serial algorithms.  Moreover, these algorithms are all convertible, so they can be used in map-reduce implementations with minimal computation cost.  The difficult part is the case of odd-length cycles, so we handle that first.
In Section~\ref{max-degree-sect}  we consider the restriction of the problem to the case where there is a degree limit for the data graph.  We show that if no node of a data graph with $m$ edges has degree higher than $\sqrt{m}$, then every connected sample graph of $p$ nodes has a serial algorithm with running time $O(m^{p/2})$; this algorithm is convertible, of course.

Throughout this section, we assume that before any of the described manipulations, the data graph has been processed to create two indexes.  One is the index discussed in Section~\ref{sg-decomp-subsect} that lets us test whether an edge exists in $O(1)$ time. This index takes $O(m)$ time to create.
The other index, also creatable in $O(m)$ time, lets us find, for each node, the set of adjacent nodes in time proportional to the degree of that node.
Since all algorithms we discuss run in time at least $O(m)$, we shall neglect the cost of index creation and use.

\subsection{Odd-Length Hamilton Cycles}
\label{odd-cycle-subsect}

To begin, we must introduce the idea of properly ordered $2$-paths. Let $\prec$ be a total order on the nodes of the data graph $G$, such that the nodes appear in nondecreasing order of their degrees. We call a $2$-path $u - v - w$ of $G$ \emph{properly ordered} if its midpoint precedes its endpoints in $\prec$, i.e. if $v \prec u$ and $v \prec w$. The following proposition shows that properly ordered $2$-paths have a $(0, 3/2)$-algorithm.

\begin{lemma}\label{pr:2-paths}
Let $G$ be data graph with $m$ edges, and let $\prec$ be a total order on the nodes of $G$ where nodes appear in nondecreasing order of degree. Then, all properly ordered $2$-paths of $G$ can be generated in $O(m^{3/2})$ time.
\end{lemma}

\begin{proof}
The proof is similar to the proofs of \cite[Theorem~3.5]{AYZ97} and \cite[Theorem~2]{Schank07}. We can generate all properly ordered $2$-paths of $G$ by considering the nodes as they appear in $\prec$, and for each node $v$, outputing $u - v - w$, for all pairs $u, w \in \Gamma_{\prec}(v)$, where $\Gamma_{\prec}(v) \equiv \{ u \in V : v \prec u \mbox{ and } (v, u) \in E\}$ is the set of $v$'s neighbors that appear after $v$ in $\prec$. The time complexity is bounded by the number of $2$-paths output by the algorithm.

We shall show that there are $O(m^{3/2})$ properly ordered $2$-paths. To this end, we call a node $v$ \emph{high degree}, if $\deg(v) \geq \sqrt{m}$, and \emph{low degree}, otherwise. We observe that there are at most $\sqrt{m}$ high-degree nodes. Thus, for each high-degree node $v$, there are at most $\sqrt{m}$ nodes in $\Gamma_{\prec}(v)$, since each member of $\Gamma_{\prec}(v)$ must itself be high-degree. Thus the number of properly ordered $2$-paths with midpoint $v$ is $O(m)$. In total, there are $O(m^{3/2})$ properly ordered $2$-paths whose midpoint is a high-degree node.
On the other hand, for each edge $e = (v, u)$, with $v \prec u$ and $v$ a low-degree node, there are at most $\sqrt{m}$ properly ordered $2$-paths $u - v - w$ that contain $e$. In total, there are $O(m^{3/2})$ properly ordered $2$-paths whose midpoint is a low-degree node.
\end{proof}

Now, let $S$ be a sample graph that contains an odd-length Hamilton cycle and possibly some additional edges, and let
$$(v_1, v_2, \ldots, v_p)$$
be any occurrence of $S$ in $G$, where the nodes are listed as they appear in the Hamilton cycle, and $v_1$ precedes $v_2, \ldots, v_p$ in $\prec$. Each such occurrence of $S$ can be decomposed into a properly ordered $2$-path $v_p - v_1 - v_2$ and $(p-3)/2$ pairs of nodes connected by an edge. Since we have a $(0, 3/2)$-algorithm for properly ordered $2$-paths and a $(0, 1)$-algorithm for edges, we can apply Lemma~\ref{a-b-lemma} and obtain the following:

\begin{theorem}\label{cor-oddenum}
Let $S$ be a $p$-node sample graph, with $p$ odd, that contains a Hamilton cycle. Then, $S$ has a $(0, p/2)$-algorithm.
\end{theorem}

\begin{proof}
Let $S$ be our $p$-node sample graph, where $p$ is odd.  Start by finding all occurrences of cycles $C_p$ in the data graph.  We do so by finding all subgraphs that are properly ordered 2-paths and combining them with $(p-3)/2$ subgraphs that are edges, using Lemma~\ref{a-b-lemma}.  Note that every cycle of length $p$ has some node that precedes all the others according to $<$, so every cycle can be constructed in this way.  Since the properly-ordered 2-path has a $(0,3/2)$-algorithm, and each of the edges has a $(0,1)$-algorithm, we can find all cycles of length $p$ with a $(0,p/2)$-algorithm.

We are not done, because we must check for each $p$-cycle that it is an instance of $S$.  That is, we must check that, in one of the $2p$ orientations of the cycle, all the edges of $S$ not on the cycle exist in the data graph.  However, these checks take time that depends only on $p$, and is in fact $O(p^3)$.  Thus, $S$ also has a $(0,p/2)$-algorithm.
\end{proof}

\begin{algorithm}[t]
\SetAlgoNoLine
\SetArgSty{normal}
\KwIn{Graph $G(V, E)$, total order $<$ on $V$, integer $k \geq 2$.}
\KwOut{Enumeration of all cycles $C_{2k+1}$ on $G$.}

\For{all $(v_1, v_2), (v_1, v_{2k+1}) \in E$
     with $v_1 \prec v_2 \prec v_{2k+1}$}{
        \For{each set of $k-1$ node-disjoint edges
        $(v_3, v_4),\ldots,(v_{2k-1},v_{2k})$\\
        not including $v_1, v_{2}, v_{2k+1}$}{
            \uIf{$v_1$ precedes $v_3,v_4,\ldots,v_{2k-1},v_{2k}$ in $\prec$}{
               \For{all permutations $(i_2, \ldots, i_k)$ of $\{ 3, 5, \ldots, 2k-1\}$ and\\ all edge orientations $b_2 b_3\cdots b_k \in \{0,1\}^{k-1}$}{
            \uIf{all edges
        $(v_2, v_{i_2+b_2}), (v_{i_2+1-b_2}, v_{i_3+b_3}), \ldots, (v_{i_{k-1}+1-b_{k-1}}, v_{i_k+b_k}), (v_{i_k+1-b_k}, v_{2k+1})$ are present in $G$}
            {output cycle $(v_1, v_2, v_{i_2+b_2}, v_{i_2+1-b_2}, \ldots,$ $v_{i_{k-1}+1-b_{k-1}}, v_{i_k+b_k}, v_{i_k+1-b_k}, v_{2k+1}, v_1)$\;}}}}}

\caption{Algorithm $\OddCycle$.}
\label{alg:odd-cycle-enum}
\end{algorithm}

Algorithm~\ref{alg:odd-cycle-enum} ({\em OddCycle})  is an implementation  of Theorem~\ref{cor-oddenum}.  It numerates all cycles $C_{2k+1}$ on a data graph $G$.

\begin{example}
Let $(v_1, v_2, v_3, v_4, v_5$, $v_6, v_7, v_1)$ be a cycle of length $7$ on data graph $G$, and assume that $v_1$ precedes $v_2, \ldots, v_7$ in $\prec$. This cycle can be uniquely decomposed into a properly ordered $2$-path $v_7 - v_1 - v_2$, with midpoint $v_1$, and a set of $2$ node-disjoint edges $(v_3, v_4)$, $(v_5, v_6)$, which do not include $v_1, v_2, v_7$ as their endpoints.
At some point, the properly ordered $2$-path $v_7 - v_1 - v_2$ is generated by the first for-loop, and the edge set $(v_3, v_4)$, $(v_5, v_6)$ is considered in the second for-loop.
%
%$\OddCycle$ verifies that $v_7 - v_1 - v_2$, $(v_3, v_4)$, and $(v_5, v_6)$ are node disjoint, and that $v_1$ preceded $v_3, v_4, v_5, v_6$ in $\prec$.
%
Then, the algorithm generates all possible permutations $(i_2, i_3)$ of the edges $(v_3, v_4)$, $(v_5, v_6)$, and all possible orientations $b_2b_3$ of the edges $(v_{i_2}, v_{i_2+1})$, $(v_{i_3}, v_{i_3+1})$.
For the permutation $(3, 5)$ and the orientation $00$, the algorithm verifies that the edges $(v_2, v_3)$, $(v_4, v_5)$, and $(v_6, v_2)$ are present in $G$, and generates cycle $(v_1, v_2, v_3, v_4$, $v_5, v_6, v_7, v_1)$.
Since $(v_1, v_2, v_3, v_4, v_5$, $v_6, v_7, v_1)$ can be uniquely decomposed as above, the algorithm generates this cycle only once.
\end{example}

Let us analyze the time complexity of $\OddCycle$.  By Prop.~\ref{pr:2-paths}, the body of the first for-loop is executed $O(m^{3/2})$ times. There are at most $m^{k-1}/(k-1)!$ different sets of $k-1$ edges considered in the second for-loop. For each such set, we can determine whether their endpoints do not include the nodes of the properly ordered $2$-path and whether the midpoint of the $2$-path is the smallest node according to $<$ in $O(k)$ time. For each set of $k-1$ node-disjoint edges, there are $(k-1)!$ permutations and $2^{k-1}$ orientations of them considered in the second for-loop. For the given properly ordered $2$-path and each such permutation and edge orientation, we can determine whether the corresponding cycle exists in $O(k)$ time. So $\OddCycle$ is a $(0, (2k+1)/2)$-algorithm, for constant any $k$.

\subsection{General Sample Graphs}

For general sample graphs, we combine  Lemma~\ref{a-b-lemma}, Example~\ref{decomp-ex}, and Theorem~\ref{cor-oddenum} with the observation that an isolated node has a $(1,0)$-algorithm.

\begin{theorem}\label{gen-enum-th}
Let $S$ be a sample graph with $p$ nodes. If $S$ can be decomposed into node-disjoint subgraphs consisting of $q$ isolated nodes, pairs of nodes connected by an edge, and
graphs with an odd-length Hamilton cycle, plus possible edges that connect nodes in two of these subgraphs, then $S$ has a $\bigl(q,(p-q)/2\bigr)$-algorithm.
\end{theorem}

Note that $q$ plus twice $(p-q)/2$ is exactly $p$, so this algorithm is always convertible.
Also, since it always pays to trade $n^2$ for $m$ in the running time, we seek for a decomposition of $S$ that minimizes the number of isolated nodes.

\subsection{Data Graphs of Bounded Maximum Degree}
\label{max-degree-sect}

The results of \cite{Alon81} imply that the running time of Theorem~\ref{gen-enum-th} is essentially best possible for general data graphs. However, if we assume an upper bound on the maximum degree of the data graph, we can obtain a stronger upper bound on the running time of the enumeration algorithm.

\begin{theorem}\label{bounded-degree-th}
Suppose data graphs are restricted to have maximum degree of at most $\Delta$ (which may be a function of the number of edges $m$), and let $S$ be a connected sample graph with $p \geq 2$ nodes. Then $S$ has an enumeration algorithm with running time of the form $O(m \Delta^{p-2})$.
\end{theorem}

\begin{proof}
The proof is by induction on $p$. For the basis, we consider a sample graph $S$ with $p = 2$ nodes. Since $S$ is connected, there must be an edge between the two nodes, and we know how to find instances of $S$ in data graph $G$ in $O(m)$ time (even without a constraint on the degree of $G$).

For the induction step, consider a sample graph $S$ with $p \geq 3$ nodes. Let $u$ be any node of $S$ that is not an articulation point, let $S_1$ consist of node $u$, and let $S_2$ be the subgraph induced by all other nodes. Since $u$ is not an articulation point, $S_2$ is a connected sample graph with $p-1 \geq 2$ nodes. By induction, $S_2$ has an algorithm with running time $O(m \Delta^{p-3})$. It therefore returns no more than $O(m \Delta^{p-3})$ instances of $S_2$. Let $u$ be connected to node $v$ in $S$ (such a node exists because $S$ is connected, and thus $u$ is not isolated). Then for each instance of $S_2$ in $G$, we map $u$ to each of $v$'s neighbors in $G$, and check which of these neighbors allow us to complete an instance of $S$ (because all the necessary edges involving $u$ exist in $G$). For each instance of $S_2$, there are at most $\Delta$ neighbors of $v$ to try as potential images of $u$, and for each of them, there are at most $p-1$ edges of $S$ to check for existence in $G$. If we index edges by their endpoints, we need $O(\Delta$) time per instance of $S_2$.  We also need to lexicographically order the nodes of $G$ so that it is possible to emit the resulting instance of $S$ only if it is lexicographically first. This can be implemented similarly to the proof of Lemma~\ref{a-b-lemma}. The total time taken is $O(m \Delta^{p-2})$.
\end{proof}

In the proof of Theorem~\ref{bounded-degree-th}, the maximum degree $\Delta$ of the data graph can be a constant or any function of $m$ and $n$. Therefore, for data graphs of constant maximum degree, all sample graphs have an $O(m)$-time enumeration algorithm, and for data graphs of maximum degree $O(\sqrt{m})$, all sample graphs have a $(0, p/2)$-algorithm. However, such an algorithm is convertible only if the maximum degree $\Delta$ of the data graph is large enough compared with the number $b$ of buckets into which we hash its nodes (e.g., if $\Delta / b = \Omega(\log n)$). Then, similarly to the proof of Theorem~\ref{comp-cost-th}, we can assume that skew is limited, and each reducer processes a subgraph with $O(m/b^2)$ edges and maximum degree $O(\Delta/b)$. Since the number of reducers is $b^p$, the total work is
$$O(b^p m/b^2 (\Delta/b)^{p-2}) = O(m\Delta^{p-2}).$$

Moreover, we note that the running-time bound of Theorem~\ref{bounded-degree-th} is essentially best possible, in the sense that for any $p \geq 2$ and any $\Delta$ sufficiently larger than $p$, a $\Delta$-regular tree with $n$ nodes contains $\Theta(n \Delta^{p-2}) = \Theta(m \Delta^{p-2})$ instances of a star with $p$ nodes. More specifically, each of the $\Theta(n/\Delta)$ nonleaf nodes of the tree is the root of $\binom{\Delta}{p-1} = \Theta(\Delta^{p-1})$ different stars with $p$ nodes.

\subsection{Joins for binary relations of different sizes}

For single binary relations, the bounds that we have given here are tight in that for any size of the relation,
the running time of our serial algorithm that enumerates all sample
graphs meets, to within a constant factor, the lower bound (given in \cite{Alon81}) hence it is
both optimal as a serial algorithm and is also a convertible algorithm. However, when we have a multiway
join over binary relations of different sizes, this is not the case as pointed out in \cite{NgoPRR12}. It is an open
question whether we can refine the bounds for this case to be more precise. Below we have done such a refinement; we give a complete
analysis for the case the sample graph is a cycle of size 5 and binary relations are of different sizes.

We are looking at the join of
$$R_1(A,B) ~JOIN~ R_2(B,C) ~JOIN~
 R_3(C,D) ~JOIN~ R_4(D,E) ~JOIN~ R_5(E,A)$$
where each relation $R_i, i=1,\ldots ,5$ has $n_i$ tuples.

\noindent{\bf Case A.}
\begin{enumerate}
 \item {\sl Condition:}  $n_1 n_5 n_3 \geq n_2 n_4$ for all cyclic automorphisms of the cycle that defines the 5-way join.

\item {\sl Complexity:}  the upper and lower bounds meet at $\sqrt{n_1 n_2 \cdots n_5}$.
\end{enumerate}

In this case the upper bound is given by the result in \cite{NgoPRR12}.
For the lower bound we construct a set of relations so that the number of sample graphs is equal to $\sqrt{n_1 n_2 \cdots n_5}$.
The size of the domain for attribute $A$ should be $\sqrt{n_1 n_5 n_3 / n_2 n_4}$
and the other attributes have analogous sizes. That is, in the
numerator are the sizes of the relations that contain $A$ (these are  $n_1$ and $n_5$)
and also the size of the relation opposite $A$ in the cycle; in this
case the relation opposite is $R_3(C,D)$ and its size is $n_3$. The
denominator is the other two relation sizes. This indeed works because: a) The product
of the sizes for the two attributes in the schema of any relation is
the size of that relation.  b) The product of all the distinct values in the attributes is
$\sqrt{n_1 n_2 \cdots n_5}$, which gives the number of all distinct sample graphs in this set of relations.

\noindent{\bf Case B.}
\begin{enumerate}
 \item {\sl Condition:} There is a condition from case A that is not satisfied, so, wlog say $n_1 n_5 n_3 \leq n_2 n_4$.

\item {\sl Complexity:}  the upper and lower bounds meet at  $n_1n_5n_3$.
\end{enumerate}

For the upper bound the algorithm now is to take the join of relations
$R_1$ and $R_5$ first and then take all combinations of the result
of the join with each tuple from relation $R_3$. For each combination we check
whether there are edges from relations $R_4$ and $R_2$ to complete the cycle.
This gives us an algorithm of complexity $n_1n_5n_3$.

We argue that
the most tuples in the result of the join will be created when
there is one value of attribute $A$ that belongs to $n_1$ tuples
of $R_1$ and also to $n_5$ tuples of $R_5$.
The argument is easy: if you split the values of
$A$ so that each value belongs to fewer tuples but still the total number of tuples in
$R_1$, $R_5$ respectively are not more than $n_1 $, $n_5 $ respectively then the total number
of tuples in the join of $R_1$ with $R_5$ only decreases.

For the lower bound, we construct a database
again but we should take several cases:
a) $ n_2> n_1 n_3$ and $ n_4> n_3n_5$. Then we assign one value to the
attribute that is shared between relations $R_1$ and $ R_5$ and
populate the relations $ R_1$, $R_5$, $R_3$ with $n_1$, $n_5$, $n_3$ tuples respectively.
Now relations $R_2$ and $R_4$ are forced to have $ n_1 n_3$  and $n_3n_5$ tuples
respectively, but it is fine because we have assumed that  $n_2 > n_1 n_3$ and $n_4 > n_3n_5$.
b) Wlog  $n_2 < n_1 n_3$. In this case we are a little more careful with the
construction. I.e., we assign only a small number of values to
attributes that belong to the relation $R_2$.  But the construction
does not present more complications because since $n_1 n_5 n_3 < n_2 n_4$,
$R_4$ will be allowed enough tuples to accommodate the number of tuples
of the relations $R_5$ and $R_3$. This is so because $n_4>n_3$ as we can deduce by
reasoning that if not, then $n_4<n_3$ combined with $n_2<n_1n_5$ gives us $n_1 n_5 n_3 > n_2 n_4$.

E.g., if $n_1=1, n_2=n, n_3=1, n_4=n,n_5=1$ then the upper and lower bound is equal to $n$.

\section{Conclusions and Open Problems}
\label{conclusions-sect}

The problem of enumerating instances of a sample graph in a huge data graph has many
applications, including social networks, threat detection, and Biomolecular networks.  We have, in this paper given algorithms that use a single round of map-reduce and are able to detect all instances of a given sample graph.  These algorithms are efficient both in communication cost between mappers and reducers and in the computation cost at the mappers and reducers.
Some interesting extensions and open problems remain.

\begin{itemize}

\item
We have not addressed the case where nodes and/or edges have labels.  Neither have we addressed the case of directed graphs.  Many of the same techniques carry over in a straightforward way.  For instance, we can still express the instances of a labeled, directed sample graph as a union of CQ's.  The automorphism groups tend to be smaller, so the number of CQ's is greater, but the same methods for evaluating CQ's by a multiway join will work.
We expect that there are provably convertible algorithms in all or almost every case, but the mapping schemes may require some thought.

\item
While we are able to minimize the number of CQ's for a given sample graph with respect to particular algorithms for generating CQ's, there may be other algorithms that will yield fewer CQ's or allow more efficient evaluation of a collection of CQ's.  We may want to consider methods other than a multiway join (and thus raise the issue of algorithms taking several rounds of map-reduce) to evaluate a collection of CQ's, especially collections with common subexpressions.

\item
Are there other restrictions on the data graph, besides limiting the degree, that yield superior convertible algorithms?

\item
To what extent does the notion of a convertible algorithm extend to other classes of map-reduce problems?
\end{itemize}